\documentclass[twoside]{article}

\usepackage{amssymb, amsmath, amsthm, mymacros, enumitem, algpseudocode}
\usepackage{fancyhdr}
\usepackage[affil-it]{authblk}
\usepackage[bookmarks=true]{hyperref}
\usepackage{tikz-cd}
\usepackage[hypcap]{caption}
\usepackage[section]{algorithm}
\floatname{algorithm}{Procedure}

\usepackage[margin=1.25in]{geometry}
\usepackage{graphicx,ctable,booktabs}

\newtheorem{thm}{Theorem}[subsection]
\newtheorem{prop}[thm]{Proposition}
\newtheorem{cor}[thm]{Corollary}
\newtheorem{lma}[thm]{Lemma}

\theoremstyle{definition}
\newtheorem{defn}[thm]{Definition}
\theoremstyle{remark}
\newtheorem{exmp}[thm]{Example}
\newtheorem{proofpart}{Part}

\newtheoremstyle{named}{}{}{\itshape}{}{\bfseries}{.}{.5em}{\thmnote{#3's }#1}
\theoremstyle{named}

\newcommand{\emodels}{\models_{\epsilon}}
\newcommand{\emodelsp}[1]{\models_{#1}}
\newcommand{\fmodels}{\vdash_\epsilon}
\newcommand{\fmodelsp}[1]{\vdash_{#1}}
\newcommand{\pmodel}[2]{(\mathcal{#1}, \mathcal{#2})}
\newcommand{\pmd}{\pmodel{M}{D}}
\newcommand{\emodelz}{\models_0}

\newcommand{\M}{\mathcal{M}}
\newcommand{\D}{\mathcal{D}}
\newcommand{\qmodels}{\models^{\mathrm{q}}}

\newcommand{\lL}{\mathcal{L}}
\newcommand{\frml}[3]{#1(\vec{#2}; \vec{#3})}
\newcommand{\phixp}{\frml \phi x p}
\newcommand{\varq}{\triangledown}
\newcommand{\eq}{\doteqdot}
\newcommand{\existseps}[1][\epsilon]{\underset{^{>#1}}\oldexists}
\newcommand{\foralleps}[1][\epsilon]{\underset{^{\ge1- #1}}\oldforall}
\newcommand{\indset}[1]{\llb #1\rrb}
\newcommand{\varqs}[2][x]{\varq_1 #1_1 \cdots \varq_{#2} #1_{#2}}
\newcommand{\Qr}[2][\epsilon]{\mathfrak{Q}^{_{#2}}_{^{#1}}}
\newcommand{\Qset}{\mathrm{QSET}}
\newcommand{\coerce}[1]{\operatorname{#1-coerce}}
\newcommand{\Ecoerce}{{\coerce{E}}}
\newcommand{\Fcoerce}{{\coerce{F}}}

\newcommand{\formone}{\underset{^{\ge 1}}\oldforall}
\newcommand{\foralle}{\underset{^{\ge \epsilon}}\oldforall}

\let\oldexists\exists
\let\oldforall\forall

\renewcommand{\powerset}{\mathfrak{P}}

\title{Computabilities of Validity and Satisfiability in Probability
Logics over Finite and Countable Models}
\author{Greg Yang
    \thanks{email: gyang@college.harvard.edu}
    \affil{Harvard University}
}

\makeatletter
\let\theauthor\@author
\let\thetitle\@title
\makeatother

\fancyhfoffset{.05\textwidth}

\begin{document}

\maketitle
\thispagestyle{empty}

\begin{abstract}
The $\epsilon$-logic (which is called $\epsilon$E-logic in this paper) of Kuyper
and Terwijn is a variant of first order logic with the same syntax, in which
the models are equipped with probability measures and in which the
$\forall x$ quantifier is interpreted as
``there exists a set $A$ of measure $\ge 1 - \epsilon$ such that for each $x
\in A$, \ldots.''
Previously, Kuyper and Terwijn proved that the general satisfiability and
validity problems for this logic are, i) for rational $\epsilon \in (0, 1)$,
respectively $\Sigma^1_1$-complete and $\Pi^1_1$-hard, and ii) for $\epsilon =
0$, respectively decidable and $\Sigma^0_1$-complete.
The adjective ``general'' here means ``uniformly over all languages.''

We extend these results in the scenario of finite models. In particular, we show
that the problems of satisfiability by and validity over finite models in
$\epsilon$E-logic are, i) for rational $\epsilon \in (0, 1)$,
respectively $\Sigma^0_1$- and $\Pi^0_1$-complete, and ii) for $\epsilon = 0$,
respectively decidable and $\Pi^0_1$-complete.
Although partial results toward the countable case are also achieved,
the computability of $\epsilon$E-logic over countable models still remains
largely unsolved.
In addition, most of the results, of this paper and of Kuyper and Terwijn,
do not apply to individual languages with a finite number of unary predicates.
Reducing this requirement continues to be a major point of research.

On the positive side, we derive the decidability of the corresponding problems
for monadic relational languages ---
equality- and function-free languages with finitely many unary and zero other
predicates.
This result holds for all three of the unrestricted, the countable, and the
finite model cases.

Applications in computational learning theory (CLT), weighted graphs, and
artificial neural networks (ANN) are discussed in the context of these
decidability and undecidability results.

\end{abstract}

\section{Introduction}
In the new age of ``big data," machine learning and statistical inference have
been increasingly applied in the technology sector, and more resources than ever
before are poured into advancing our understanding of these techniques.
One approach to this end is to reconcile the inductive nature of machine
learning with the deductive discipline of logic.
Previous attempts include one by computer scientist Leslie Valiant
, the creator of the PAC (Probably Approximately
Correct) model of computational learning theory (CLT).
In \textit{Robust Logics} \cite{Valiant_robust_logics}, he tried to combine the
PAC model with a fragment of first order logic (FOL) in the context of finite
models.
The logician H. Jerome Keisler in \cite{keisler_p_quantifiers} also investigated
a variant of FOL with probabilistic quantifiers of the form $(P x
\ge r)$ meaning ``holds for $x$ in a set of measure at least $r$.''

Most recently, Terwijn and Kuyper \cite{Terwijn_intro}
invented a probability logic with a fixed error parameter, called
$\epsilon$-logic (or, in this paper, $\epsilon$E-logic), that is inspired by
features of both Valiant and Keisler's work.
This $\epsilon$E-logic uses the same syntax as FOL and differ only
in that 1) the models are given probability measures, and 2) the $\forall$
quantifier has the interpretation of ``holds for $x$ in a set of measure at
least $1 - \epsilon$.''
In particular, the $\exists$ quantifier keeps the same, non-probabilistic
interpretation as in first order logic.

Such an unusual, asymmetric definition was motivated by the key property of
$\epsilon$E-logic to be learnable through examples, in a sense related to
Valiant's PAC-learning model \cite[Thm.~2.3.3]{probability_logic}:
roughly, for any desired error bound
$\epsilon$ and an example oracle that emits elements of the universe $M$
according to a distribution $\D$, we can learn
in time polynomial in $\f 1 \epsilon$ whether $\pmd \emodels\phi$ or $\pmd
\emodels \neg \phi$.
Thus, $\epsilon$E-logic has an \textit{inductive} property, in addition to
promises of \textit{deductive} properties that would seem to carry over from
classical first order logic.

However, it turns out that deductive reasoning in $\epsilon$E-logic is
computationally much harder than first order logic in the general case.
In fact its complexity does not even reside in the arithmetic hierarchy, but
rather the analytic one.
As a result, there is no algorithm to decide (uniformly over all first order
languages) whether a given sentence is valid or satisfiable.
The following table summarizes the current knowledge on the
satisfiability and validity
\footnote{This notion of validity is over all
probability models.
It is called \textit{normal $\epsilon$E-validity} in this paper.
See definition (\ref{def_etau}) and \cite[remark before
thm~2.6]{sat_is_sigma_1_1}}
problems of $\epsilon$E-logic (defined below in (\ref{def_esat}) and
(\ref{def_etau})).
Each tuple $\la n_1, n_2, \ldots\ra$ represents that the corresponding result
requires the language to have at least $n_1$ unary predicates, at least $n_2$
binary predicates, and so on.
In particular, the value of $\omega$ means an infinite number of the
corresponding type of predicate is necessary.
The empty tuple $\la \ra$ means that any language suffices.
Finally $\la*\ra$ denotes that the set of $\epsilon$E-valid sentences coincides
with the set of valid sentences in first order logic, so any language
that admits a decision algorithm in FOL will do so in $\epsilon$E-logic as well.
As general FOL validity is $\Sigma^0_1$-complete, this means that
general $\epsilon$E-validity is also $\Sigma^0_1$-complete.

\begin{table}[h]
    \centering
    \begin{tabular}{ c|c c }
        \hline\hline
                    &	$\epsilon \in (0, 1) \cap \Q$	&	$\epsilon = 0$	\\
        \hline
        $\epsilon$E-satisfiability	&
            $\Sigma^1_1$-complete	$\la \omega, 6, 2 \ra$
                \cite[thm~7.6]{sat_is_sigma_1_1}&
                decidable	$\la\ra$	\cite[thm~6.7]{sat_is_sigma_1_1}\\
        $\epsilon$E-validity &
            $\Pi^1_1$-hard		$\la \omega, 3, 2\ra$
                   \cite[thm~4.2]{val_is_pi_1_1}&
                $\Sigma^0_1$-complete	$\la*\ra$
                    \cite[prop~3.2]{Terwijn_dec}\\
        \hline\hline
    \end{tabular}
    \caption{Current knowledge on general $\epsilon$E-satisfiability
    and $\epsilon$E-validity \cite[Table 1]{sat_is_sigma_1_1}}
    \label{tble:general_computability}
\end{table}

It is still open, however, whether a smaller fragment of $\epsilon$E-logic
--- for example, languages with a finite number of unary predicates, or
languages with a single binary predicate ---
admits an easier complexity.

Notice that the results are not symmetrical as in the case of first order logic,
where $\phi$ is valid iff $\neg \phi$ is not satisfiable.
Indeed, $\epsilon$E-logic is
\textit{paraconsistent} \cite[Prop.~2.2.1]{probability_logic}, because it's
possible for $\pmd \emodels \forall x \phi(x)$ and $\pmd \emodels \exists x \neg
\phi(x)$ to hold at the same time (for example if the set of $x$ satisfying
$\neg \phi$ has measure 0 but is not empty).

In this paper, we answer the counterpart questions for
$\epsilon$E-satisfiability by and validity over finite models and in some
cases, countable models (see
(\ref{def:finite_countable_validity_satisfiability}) for definitions).
As noted above, we cannot in general answer the satisfiability question by just
answering the validity question, nor vice versa.

In first order logic, Trachtenbrot's theorem \cite{libkin} asserts that,
perhaps counterintuitively, assessing the validity of a theorem over
only finite models is $\Pi^0_1$-hard.
Therefore we do not even have a deductive calculus for this task.

We show that other than the case of $\epsilon = 0$, Trachtenbrot's theorem
holds also for $\epsilon$E-logic.
More precisely, we will establish the following characterizations.

\begin{table}[h]
    \centering
    \begin{tabular}{ c|c c }
        \hline\hline
                    &	$\epsilon \in (0, 1) \cap \Q$	&	$\epsilon = 0$	\\
        \hline
        finite $\epsilon$E-satisfiability
            &	$\Sigma^0_1$-complete $\la \omega, 3\ra$
                (\ref{thm:Esat_sigma1_complete})
                &	decidable	$\la\ra$
                    (\ref{thm:0F-validity_decidable})\\
        finite $\epsilon$E-validity
            &$\Pi^0_1$-complete $\la \omega, 1\ra$
                (\ref{thm:Eval_Pi1_complete})
                & $\Pi^0_1$-complete $\la*\ra$
                    (\ref{thm:Eval_Pi1_complete})\\
        \hline\hline
    \end{tabular}
    \caption{Summary of results in this paper: the finite model case}
    \label{tble:finite_computability}
\end{table}

Here, the tuple notation $\la n_1, n_2, \ldots\ra$ denotes language
requirements, as in the last table, but $\la*\ra$ means that the set of
$\epsilon$E-valid sentences over \textit{finite models} coincides with the
corresponding set \textit{over finite models} with regard to FOL.

Like in the unrestricted case, it is still open whether the conditions
on signature can be significantly weakened.

Hence, other than the case of finite 0E-satisfiability, the
$\epsilon$E-satisfiability and $\epsilon$E-validity problems in the finite
model scenario are as hard as the corresponding problems in ordinary first order
logic \cite[p. 166]{libkin}.
Therefore, no general deduction mechanism exists for theorems over finite
models.

In contrast, for $\epsilon \in [0, 1)$ rational, we show that these
problems are decidable over monadic relational languages, which are
languages with only unary predicates and no function symbols or equality.
This mirrors the characterizations for the corresponding FOL fragments.

Despite these successes, the countable case remains largely unsolved.
While models of $\epsilon$E-logic have a downward Lowenheim-Skolem theorem
transforming them into equivalent continuum-sized models
\cite[Thm.~4.6]{model_theory_of_e_logic},
this theorem does not hold when ``continuum-sized'' is replaced with
``countable.'' So unlike FOL, the set of sentences $\epsilon$E-valid over all countable
models does not coincide a priori with the set of those over all models.

In the format of the previous tables, we summarize the current knowledge on
countable $\epsilon$E-validity and satisfiability in table
(\ref{tble:countable_computability}).

\begin{table}[h]
    \centering
    \begin{tabular}{ c|c c }
        \hline\hline
                    &	$\epsilon \in (0, 1) \cap \Q$	&	$\epsilon = 0$	\\
        \hline
        countable $\epsilon$E-satisfiability
            &	unknown
                &	decidable	$\la\ra$
                    (\ref{thm:0F-validity_decidable})\\
        countable $\epsilon$E-validity
            &$\Sigma^0_1$-hard $\la \omega, 1\ra$
                (\ref{cor:countabl_epsval_Sigma1_hard})
                & $\Sigma^0_1$-complete $\la*\ra$
                    (\ref{cor:countabl_0val_Sigma1_complete})\\
        \hline\hline
    \end{tabular}
    \caption{Summary of results in this paper: the countable case}
    \label{tble:countable_computability}
\end{table}

The outline of the paper is thus:
After introducing some notations and prerequisites, we review the
basic definitions of $\epsilon$E-logic, $\epsilon$-model, the validity and
satisfiability problems, and other related notions in section
(\ref{sec:probability_logics}).
We then define in section (\ref{ssec:flogic}) the dual logic,
\textit{$\epsilon$F-logic}, in whose terms we phrase many of our results.
Briefly, in $\epsilon$F-logic, the syntax is once again identical to that of
FOL, but the quantifier $\exists x$ is interpreted to mean
``there exists a set $A$ of measure $> \epsilon$ such that for each $x \in A$,
\ldots.''
Having laid out the analogue concepts over finite models,
we dive straight into examples and applications in section
(\ref{ssec:examples_and_applications}), hoping to motivate our main theorems
and future research.

There, we employ $\epsilon$E-logic and $\epsilon$F-logic to
\begin{itemize}
    \item model the \textit{approximation concept existence} assumption
    of PAC learning;
    \item develop rudimental theories of graphs with weighted vertices and
    graphs with weighted edges; and
    \item compute the \textit{linear threshold update rule} in artificial neural
    networks.
\end{itemize}

With these examples in mind, we begin our deductions.
In section (\ref{ssec:0Fval}), we prove the decidability of 0E-satisfiability
over both finite models and countable models.
In section (\ref{ssec:monadic}), we show that the satisfiability and validity of
sentences in any monadic relational language both reduce to linear programming,
and thus are decidable whenever $\epsilon \in \Q$.
During the development of the reduction, we introduce trees as semantic tools.
These ideas are extended in section (\ref{ssec:q-sentence}) to define the
semantics of \textit{q-sentences}, which generalize both $\epsilon$E- and
$\epsilon$F-logics by allowing all forms of quantifiers.
This new form of sentences allows us to rigorously state the inter-reduction
results of Kuyper and Terwijn between different rational $\epsilon$ parameters.
Equipped aptly with powerful machinery, in section (\ref{ssec:Esat})
we tackle the $\Sigma^0_1$-completeness
of the $\epsilon$E-satisfiability problem over finite models.
The hardness proof involves a painstaking encoding of the halting set in
a suitable language.
The definability proof utilizes a perturbation lemma that simplifies
satisfiability to that by models with rationally-valued distributions.
Last but not least, we show in section (\ref{ssec:eval}) that 0E-validity
coincides with FOL validity when both are restricted to finite or countable
models.
(Of course, in the countable case, the ``countably FOL valid'' sentences are
just the unrestricted valid sentences by Lowenheim-Skolem).
With Kuyper's inter-reduction theorem, we deduce that $\epsilon$E-validity
for finite models is $\Pi^0_1$-complete and that for countable models is
$\Sigma^0_1$-hard for any strong enough language.
Finally, we wrap up and mention possible future directions of research in
section (\ref{sec:future}).

\subsection{Notation and Prerequisites}

Sets of the form $\{1, 2, \ldots, n\}$ will be abbreviated $\indset n.$
\subsubsection{Logic}

Script upper case letters $\M$ and $\mathcal N$ are used to denote first order
models.
Their underlying universes are written $M$ and $N$.
All first order languages are assumed to have an at most countable signature.

$\vec x$ denotes a finite sequence of variables or parameters. $|\vec x|$
denotes the length of this sequence.
$\phi(\vec x, \vec y)$ will always represent a formula with free variables $x_1,
\ldots, x_n$ and $y_1, \ldots, y_m$, possibly with other bound variables.
In the context of a first order model $\M$,
    $$\phi(\vec x, \vec y; \vec p) =
        \phi(x_1, \ldots, x_n, y_1, \ldots, y_m; p_1, \ldots, p_n),$$
represents a formula with free variables $\vec x, \vec y$ and parameters $p_i
\in M$.
$\forall \vec x$ is the shorthand for the quantifier block $\forall x_1 \forall
x_2 \cdots \forall x_n$.
Similarly, $\exists \vec x$ is the shorthand for the quantifier block $\exists
x_1 \exists x_2 \cdots \exists x_n$.

If $\phi$ is a (formal or informal) sentence, then $\|\phi\| \in \{0, 1\}$
denotes its truth value.

In formulas, we adopt the convention that $\land$ is parsed before $\lor$ when
written without parentheses. For example,
    $$A(x, y) \land B(y, z) \lor R(z) \land x = z$$
is parsed as
    $$(A(x, y) \land B(y, z)) \lor (R(z) \land x = z).$$
Formulas of the form
    $$\phi_1 \land \phi_2 \land \cdots \land \phi_k \limplies \psi$$
are parsed as
    $$(\phi_1 \land \phi_2 \land \cdots \land \phi_k) \limplies \psi.$$

In addition, we will use square brackets [] in place of parentheses () when
doing so improves the readability.

A subset $A \sbe \N$ is called $\Sigma^0_1$ or $\Sigma^0_1$-definable if
$A$ is the range of a recursive function. A subset $A \sbe \N$ is called
$\Sigma^0_1$-hard if for every $\Sigma^0_1$ set $A'$, there is a computable
many-one reduction from $A'$ to $A$. A subset $C \sbe \N$ is called
$\Sigma^0_1$-complete if $C$ is both $\Sigma^0_1$-definable and
$\Sigma^0_1$-hard.

Dually, a subset $B \sbe \N$ is called $\Pi^0_1$ or $\Pi^0_1$-definable (resp.
$\Pi^0_1$-hard and resp. $\Pi^0_1$-complete) if its complement $\N - B$ is
$\Sigma^0_1$ (resp. $\Sigma^0_1$-hard and resp. $\Sigma^0_1$-complete).

Please refer to Soare \cite{Soare} for unexplained concepts in
computability.

\subsubsection{Measure Theory}

Let $X$ be a set. $\powerset(X)$ denotes the power set of $X$.
By \textit{a measure $\mu$ on $X$}, we mean a set function $\mu: \Aa \to [0, \infty]$
that is defined and countably additive on some $\sigma$-algebra $\Aa \sbe
\powerset(X)$. The triple $(X, \Aa, \mu)$ is called a \textit{measure space.}
Similarly, by \textit{a finitely additive measure $\mu$ on $X$}, we mean a set
function $\mu: \Aa \to [0, \infty]$ defined on some Boolean algebra $\Aa \sbe
\powerset(X)$, and $\mu$ is finitely additive on $\Aa$.
In both cases, the $\sigma$-algebra or Boolean algebra $\Aa$ on which $\mu$ is
defined is denoted $\dom \mu$.
We say that \textit{$\mu$ is everywhere defined} if $\dom \mu = \powerset(X)$.
We say that \textit{a measure $\mu$ on $X$ is
extended by $\mu'$} if $\dom \mu' \spe \dom \mu$.
The $\mu$-measure of a set of elements satisfying some condition $\phi$ is
denoted
    $$\mu(x: \phi(x)).$$
When we say ``$A$ has measure at least $1/2$,'' we implicitly assume $A$ is
first of all measurable, and then that it has measure at least $1/2$.

Primarily we will be discussing \textit{probability measures}, i.e. measures
$\mu$ on $X$ with $\mu(X) = 1$.
We use upper case script letters starting
from $\D, \mathcal E$, etc to name them.
In the contexts of probability measures, we will also use the probability
notation $\Pr_{x \in \D}[\phi(x)]$ interchangeably with $\D(x: \phi(x))$.

Please refer to Bogachev \cite{Bogachev} for unexplained concepts in measure
theory.

\subsubsection{Linear Programming}

A \textit{linear program} is a triple $L = (\vec x, E, f)$ where
\begin{itemize}
\item $\vec x$ is a set of $|\vec x| = k$ variables,
\item $E$ is a set of $|E| = n$ (weak) linear inequalities
    $$e^i: \sum_{j = 1}^{k} c_j^i x_j \ge d^i,$$
    and
\item $f$ is a linear function, called the \textit{object function},
    $$f(\vec x) = \sum_{j=1}^{x} q_j x_j.$$
\end{itemize}

In general, if assignment $\vec x = \vec a$ satisfies all inequalities $e_i$,
then we write
    $$C\vec x \ge \vec d.$$
Here $C$ is the matrix $\{c_j^i\}_{i, j}$ with row vectors
$\vec{c^i} = (c_1^i, \ldots, c_k^i)$, and
$\vec d = (d^1, \ldots, d^n)$.

Because each equality $\sum p_j x_j = r$ can be written as two weak
inequalities, we can allow $E$ to contain equations as well.

To \textit{maximize} $L$ is to find
    $$\max(L) := \max_{C\vec x \ge \vec d} f(\vec x).$$
Likewise, to \textit{minimize} $L$ is to find
    $$\min(L) := \min_{C\vec x \ge \vec d} f(\vec x).$$

$L$ is said to be \textit{feasible} if $\{\vec x: C\vec x \ge\vec d\}$ is
nonempty.
In other words, $L$ is feasible iff $\max(L) > - \infty$ iff $\min(L) < \infty$.

In this paper, we are mainly concerned with the feasibility problem of linear
programs.
As such, we conveniently identify each program $L$ with its set of inequalities
$E$.

In the \textit{arithmetic model of computation}, the arithmetic operations of
addition, multiplication, subtraction, division, and comparison are assumed to
take unit time.
It is known that maximizing, minimizing, and finding the feasibility of a linear
program is polynomial time in the arithmetic model \cite{LP_complexity_survey}.
Because all such arithmetic operations on rational numbers are decidable,

\begin{prop} \label{prop:LPrat_decidable}
The feasibility problem of linear programs with rational coefficients is
decidable.
\end{prop}

We will also briefly cross path with \textit{strict linear programs}.
These are linear programs $L = (\vec x, E, f)$ where $E$ consists of strict
inequalities or equalities but with no weak inequalities.

Please refer to Schrijver \cite{schrijver} for unexplained concepts with regard
to linear programming.

\section{Probability Logics}
\label{sec:probability_logics}

Here we formalize the ideas touched in the introduction. First, we revisit the
definitions of $\epsilon$E-logic, $\epsilon$E-model, and related ideas as
defined by Terwijn and Kuyper.
Then, we introduce $\epsilon$F-logic, the dual logic of
$\epsilon$E-logic.
Finally, an abundance of examples and applications are provided to clarify these
ideas.

\subsection{E-logic}
Let
\begin{itemize}
    \item $\lL$ be a first order language, possibly containing equality, of
a countable signature;
    \item $\epsilon \in [0, 1]$;
    \item $\M$ be a first-order model with universe $M$;
    \item $\D$ be a probability measure on $M$ defined on some
    $\sigma$-algebra $\dom \D \subseteq \powerset(M)$.
\end{itemize}

\begin{defn}[$\epsilon$E-truth] \label{def:E-truth}
\footnote{adapted from \cite{sat_is_sigma_1_1}}
Let $\phixp =
\phi(x_1, \ldots, x_n; p_1, \ldots, p_n)$ represent a formula with variables
$\vec x$ and parameters $p_i \in \M$. We define the notion of
\textbf{$\epsilon$E-truth}, denoted $\pmd \emodels \phi$, inductively as follows:

\begin{enumerate}
  \item For every atomic formula $\phi(\vec x; \vec p)$:
  $$\pmd \emodels \phixp \iff  \M \models \phixp.$$
  That is, for \textit{all} tuples $(a_1, a_2, \ldots, a_n) \in \M$, $\phi(\vec
  a; \vec p)$ holds.
  \item We treat the logical connectives $\land$ and $\lor$ classically. For
  example, for $\vec x, \vec y, \vec z$ distinct sequences of variables,
  $$\pmd \emodels \phi(\vec x, \vec z; \vec p) \land \psi(\vec y, \vec z; \vec
  p) $$
  iff
  for all $\vec a \in M^{|\vec x|}, \vec b \in M^{|\vec y|}, \vec c \in M^{|\vec
  z|}$,
  $$\pmd \emodels \phi(\vec a, \vec c; \vec p) \land \psi(\vec b, \vec c; \vec
  p)$$
  \item The existential quantifier is treated classically:
  $$\pmd \emodels \exists \vec x \phi(\vec x, \vec y; \vec p)$$
  iff there exists $\vec a \in M^{|\vec x|}$ such that
  $$\pmd \emodels \phi(\vec a, \vec y; \vec p).$$
  \item The universal quantifier is interpreted probabilistically:
  $$ \pmd \emodels \forall x \phi(x, \vec y; \vec p) \iff
  \Pr_{a \sim \D}[\pmd \emodels \phi(a, \vec y; \vec p)]
  \ge 1 - \epsilon.$$
  Note that the universal quantifier in this definition binds a single variable
  $x$; in general it's \textit{not} true that
  $$ \pmd \emodels \forall \vec x \phi(\vec x, \vec y; \vec p) \iff
  \Pr_{\vec a \sim \D^{|\vec x|}}[\pmd \emodels \phi(\vec a, \vec y; \vec p)]
  \ge 1 - \epsilon.$$
  \item The case of negation is split in subcases as below:
  \begin{enumerate}
    \item For $\phi$ atomic, $\pmd \emodels \neg \phixp \iff \pmd \not
    \emodels \phixp$.
    \item $\neg$ distributes classically over $\land$ and $\lor$, e.g.
    $$\pmd \emodels \neg(\phi(\vec x, \vec z; \vec p) \land \psi(\vec y, \vec
    z; \vec p)) \iff \pmd \emodels \neg \phi(\vec x, \vec z; \vec p) \lor \neg
    \psi(\vec y, \vec z; \vec p).$$
    \item $\pmd \emodels \neg \neg \phixp \iff \pmd \emodels \phixp$.
    \item $\pmd \emodels \neg \exists x \phi(x, \vec y; \vec p) \iff \pmd
    \emodels \forall x\neg \phi(x, \vec y; \vec p)$.
    \item $\pmd \emodels \neg \forall x \phi(x, \vec y; \vec p) \iff \pmd
    \emodels \exists x \neg \phi(x, \vec y; \vec p)$.
  \end{enumerate}
  \item The implication symbol $\limplies$ reduces to boolean combinations
  classically:
  $$\pmd \emodels \phi(\vec x, \vec z; \vec p) \limplies \psi(\vec y, \vec z; \vec
  p)
  \iff
  \pmd \emodels \neg \phi(\vec x, \vec z; \vec p) \lor \psi(\vec y, \vec z;
  \vec p).$$
  \item
  The equivalence symbol $\lequiv$ reduces to the conjunction of two
  implications:
  $$\pmd \emodels \phi(\vec x, \vec z; \vec p) \lequiv \psi(\vec y, \vec z; \vec
  p)$$
  $$\text{iff}$$
  $$\pmd \emodels [\phi(\vec x, \vec z; \vec p) \limplies \psi(\vec y, \vec z;
  \vec p)] \land [\psi(\vec x, \vec z; \vec p) \limplies \phi(\vec y, \vec z;
  \vec p)]$$

\end{enumerate}

This logic system is called \textbf{$\epsilon$E-logic}.
When referring to the set of all such logics for $\epsilon \in [0, 1]$ or when
$\epsilon$ is a fixed parameter implicit in the context, we simply use the term
\textbf{E-logic}.
\end{defn}

To make sure that the $\forall$ quantifier makes sense, we need to impose
measurability conditions on definable sets. In this paper, \textit{classical
models} refer to the models used in ordinary first-order logic. They are
distinct from the concept defined here:
\begin{defn}
Let $\lL$ be a first order language of a countable signature, possibly
containing equality, and let $\epsilon \in [0, 1]$. Then an
\textbf{$\epsilon$E-model} for the language $\lL$ consists of a classical
first-order $\lL$-model $\M$ together with a probability measure $\D$ over $\M$
such that:
\begin{enumerate}
  \item For all formulas $\phi = \phi(x_1, \ldots, x_n)$ and all $a_1, \ldots,
  a_{n-1} \in \M$, the set
  $$\{a_n \in \M: \pmd \emodels \phi(a_1, \ldots, a_n)\}$$
  is $\D$-measurable. \label{definable=>measurable}
  \item All relations of arity $n$ are $\D^n$-measurable (including equality, if
  it is in $\lL$), and all functions of arity $n$ are measurable as functions
  from $(\M^n, \D^n)$ to $\pmd$. In particular, constants are $\D$-measurable.
\end{enumerate}
A \textbf{probability model} is a pair $\pmd$ that is an $\epsilon$E-model for
every $\epsilon \in [0, 1]$.
\label{def_emodel}
\end{defn}

\begin{defn}
Two $\epsilon$E-models $\pmd$ and $\pmodel N E$ are
\textbf{$\epsilon$-elementarily equivalent}, denoted by $$\pmd \equiv_\epsilon
\pmodel N E,$$ iff for every formula $\phi$, $$\pmd \emodels \phi
\iff \pmodel N E \emodels \phi.$$\label{def:model_equivalence}
\end{defn}

\begin{defn}
Two formulas $\phi$ and $\psi$ are
\textbf{$\epsilon$-equivalent}, denoted by $$\phi \equiv_\epsilon \psi$$ iff for
every $\epsilon$E-model $\pmd$, $$\pmd \emodels \phi \iff \pmd \emodels \psi.$$

$\phi$ and $\psi$ are called \textbf{(semantically) equivalent}, written

$$\phi \equiv \psi,$$
if $\phi \equiv_\epsilon \psi$ for all $\epsilon \in [0, 1]$.
\label{def:formula_equivalence}
\end{defn}

\begin{defn}
Let $\phi$ be a first order sentence. We say that $\phi$ is
    \textbf{$\epsilon$E-valid}
if for all $\epsilon$E-models $\pmd$, $\pmd \emodels \phi$.
$\phi$ is
    \textbf{normally $\epsilon$E-valid}
iff for all \textit{probability models} $\pmd$, $\pmd \emodels \phi$.
\label{def_etau}
\end{defn}

Similarly we define
\begin{defn}
A sentence $\phi$ is said to be \textbf{$\epsilon$E-satisfiable} if there exists
an \textit{$\epsilon$E-model} $\pmd$ such that $\pmd \emodels \phi$.
\label{def_esat}
\end{defn}

To distinguish between these concepts and the analogue concepts over finite
and countable $\epsilon$E-models, we also prefix these terms with
\textit{unrestricted} or postfix them with \textit{in the unrestricted case}.
For example, $\epsilon$E-valid means the same thing as
\textit{unrestricted $\epsilon$E-valid}, or as
\textit{$\epsilon$E-valid in the unrestricted case}.

Likewise, to make the distinction clear from the corresponding notions in FOL,
we say
\begin{defn}
A formula $\phi$ is \textbf{classically valid} if every first order model
satisfies $\phi$.

A formula $\phi$ is \textbf{classically satisfiable} if some first order model
satisfies $\phi$.
\end{defn}

We could define \textit{normally $\epsilon$E-satisfiability} in analogy to
normally $\epsilon$E-validity, but this concept would be
equivalent to $\epsilon$E-satisfiability \cite[Thm~2.6, Prop.~2.7]{sat_is_sigma_1_1}.

Finally, we record the following proposition which will often be implicitly
applied.

\begin{prop}[Terwijn \cite{Terwijn_intro}]
Every formula $\phi$ is semantically equivalent to a formula $\phi'$ in prenex
normal form. \label{prenex}
\end{prop}

%
%
%
%

\subsection{The Dual Logic, F-logic}
\label{ssec:flogic}

\begin{defn}
Let $\lL$ be a countable first order language, possibly containing equality.
Let $\Phi = \Phi(x_1, \ldots, x_n)$ be a first order
formula in the language $\lL$, and let $\epsilon \in [0, 1]$. If $\pmd$ is an
$\epsilon$E-model, then we define \textbf{$\epsilon$F-truth}, written $\pmd
\fmodels \Phi$, by
$$\pmd \fmodels \Phi \iff \pmd \not \emodels \neg \Phi.$$

We call the logic under $\fmodels$ \textbf{$\epsilon$F-logic}.
\end{defn}

Suppose for quantifiers $\varq_1, \varq_2, \ldots, \varq_n \in \{\forall,
\exists\}$ and a quantifier free formula $\psi$,
$$\Phi(\vec y; \vec p) := \varq_1 x_1 \varq_2 x_2 \cdots \varq_n x_n
\psi(x_1, x_2, \ldots, x_n, \vec y; \vec p).$$
Then if $\varq'_i$ denotes the dual quantifier of $\varq_i$
(interchange $\exists$ with $\forall$),
$$\neg\Phi(\vec y; \vec p) \equiv_\epsilon \varq'_1 x_1
\varq'_2 x_2 \cdots \varq'_n x_n \neg\psi(x_1, x_2, \ldots, x_n, \vec y;
\vec p).$$

For example, if $\varq_i$ is $\forall$ for odd $i$ and $\exists$ for even $i$,
and $n$ is odd, then $\pmd \not \emodels \neg \Phi$ iff
\begin{quote}
for all $x_1$,

there exists a set of $x_2$ with measure strictly greater than $\epsilon$ such
that,

for all $x_3$,

$\vdots$

for all $x_n$,

$\psi(x_1, x_2, \ldots, x_n)$ holds.
\end{quote}

With this remark, it's easy to see that $\epsilon$F-logic is the \textit{dual
logic} of $\epsilon$E-logic, in that the quantifier $\forall$ is interpreted
classically, while the quantifier $\exists$ is interpreted as ``with measure
strictly greater than $\epsilon$.''

More formally and in parallel with the inductive definition given for E-logic, we
can write

\begin{enumerate}
  \item For every atomic formula $\phi(\vec x; \vec p)$:
  $$\pmd \fmodels \phixp \iff  \M \models \phixp.$$
  That is, for \textit{all} tuples $(a_1, a_2, \ldots, a_n) \in \M$, $\phi(\vec
  a; \vec p)$ holds.
  \item We treat the logical connectives $\land$ and $\lor$ classically. For
  example, for $\vec x, \vec y, \vec z$ distinct sequences of variables,
  $$\pmd \fmodels \phi(\vec x, \vec z; \vec p) \land \psi(\vec y, \vec z; \vec
  p) $$
  iff
  for all $\vec a \in M^{|\vec x|}, \vec b \in M^{|\vec y|}, \vec c \in M^{|\vec
  z|}$,
  $$\pmd \fmodels \phi(\vec a, \vec c; \vec p) \land \psi(\vec b, \vec c; \vec
  p)$$
  \item The universal quantifier is treated classically:
  $$\pmd \fmodels \forall \vec x \phi(\vec x, \vec y; \vec p)$$
  iff for all $\vec a \in M^{|\vec x|}$
  $$\pmd \fmodels \phi(\vec a, \vec y; \vec p).$$
  \item The existential quantifier is interpreted probabilistically:
  $$\pmd \fmodels \exists x \phi(x, \vec y; \vec p) \iff
  \Pr_{a \sim \D}[\pmd \fmodels \phi(a, \vec y; \vec p)]
  > \epsilon.$$
  Note that the existential quantifier in this definition binds a single
  variable $x$; in general it's \textit{not} true that
  $$ \pmd \fmodels \exists \vec x \phi(\vec x, \vec y; \vec p) \iff
  \Pr_{\vec a \sim \D^{|\vec x|}}[\pmd \fmodels \phi(\vec a, \vec y; \vec p)]
  > \epsilon.$$
  \item The case of negation is split in subcases as below:
  \begin{enumerate}
    \item For $\phi$ atomic, $\pmd \fmodels \neg \phixp \iff \pmd \not
    \fmodels \phixp$.
    \item $\neg$ distributes classically over $\land$ and $\lor$, e.g.
    $$\pmd \fmodels \neg(\phi(\vec x, \vec z; \vec p) \land \psi(\vec y, \vec
    z; \vec p)) \iff \pmd \fmodels \neg \phi(\vec x, \vec z; \vec p) \lor \neg
    \psi(\vec y, \vec z; \vec p).$$
    \item $\pmd \fmodels \neg \neg \phixp \iff \pmd \fmodels \phixp$.
    \item $\pmd \fmodels \neg \exists x \phi(x, \vec y; \vec p) \iff \pmd
    \fmodels \forall x\neg \phi(x, \vec y; \vec p)$.
    \item $\pmd \fmodels \neg \forall x \phi(x, \vec y; \vec p) \iff \pmd
    \fmodels \exists x \neg \phi(x, \vec y; \vec p)$.
  \end{enumerate}
  \item The implication symbol $\to$ reduces to boolean combinations
  classically:
  $$\pmd \fmodels \phi(\vec x, \vec z; \vec p) \limplies \psi(\vec y, \vec z; \vec
  p)
  \iff
  \pmd \fmodels \neg \phi(\vec x, \vec z; \vec p) \lor \psi(\vec y, \vec z;
  \vec p).$$
  \item
  The equivalence symbol $\lequiv$ reduces to the conjunction of two
  implications:
  $$\pmd \fmodels \phi(\vec x, \vec z; \vec p) \lequiv \psi(\vec y, \vec z; \vec
  p)$$
  $$\text{iff}$$
  $$\pmd \fmodels [\phi(\vec x, \vec z; \vec p) \limplies \psi(\vec y, \vec z;
  \vec p)] \land [\psi(\vec x, \vec z; \vec p) \limplies \phi(\vec y, \vec z;
  \vec p)]$$
\end{enumerate}

We can similarly define \textit{$\epsilon$F-models} for F-logic by replacing
condition \ref{definable=>measurable} of definition (\ref{def_emodel}) with
\begin{itemize}
  \item For all formulas $\phi = \phi(x_1, \ldots, x_n)$ and all $a_1, \ldots,
  a_{n-1} \in \M$, the set
  $$\{a_n \in \M: \pmd \fmodels \phi(a_1, \ldots, a_n)\}$$
  is $\D$-measurable.
\end{itemize}

However, note that
$$\{a_n \in \M: \pmd \fmodels \phi(a_1, \ldots, a_n)\} =
M - \{a_n \in \M: \pmd \emodels \neg\phi(a_1, \ldots, a_n)\}$$
and thus, \textit{$\pmd$ is an $\epsilon$E-model iff it's also an
$\epsilon$F-model}. Henceforward we will uniformly adopt the term
\textbf{$\epsilon$-model} for this use case.

Similarly, it should be clear from definitions (\ref{def:model_equivalence}) and
(\ref{def:formula_equivalence}) that $\pmd \equiv_\epsilon \pmodel N E$ iff
$$ \pmd \fmodels \phi \iff \pmodel N E \fmodels \phi,$$
and $\phi \equiv_\epsilon \psi$ iff for every $\epsilon$-model $\pmd$,
$$\pmd \fmodels \phi \iff \pmd \fmodels \psi.$$

\textbf{$\epsilon$F-validity} and \textbf{$\epsilon$F-satisfiability} (along
with their synonyms) are defined similar to definitions (\ref{def_etau}) and
(\ref{def_esat}).
Due to duality, we have that

\begin{prop}
\newcommand{\sava}{\bigcirc}
$\phi$ is $\epsilon$F-valid iff $\neg \phi$ is $\epsilon$E-satisfiable.
In general, $\phi$ is $\epsilon$X-$\sava$ iff $\neg \phi$ is
not $\epsilon$Y-$\sava'$, where (X, Y) is a permutation of
\{E, F\}, and $(\sava, \sava')$ is a permutation of
\{valid, satisfiable\}.
\label{fereduction}
\end{prop}

Table (\ref{tble:F_general_computability}) states the dual version of
table (\ref{tble:general_computability}).

\begin{table}[h]
    \centering
    \begin{tabular}{ c|c c }
        \hline\hline
                    &	$\epsilon \in (0, 1) \cap \Q$	&	$\epsilon = 0$	\\
        \hline
        $\epsilon$F-validity	&
            $\Pi^1_1$-complete	$\la \omega, 6, 2 \ra$
                \cite[thm~7.6]{sat_is_sigma_1_1}&
                decidable	$\la\ra$	\cite[thm~6.7]{sat_is_sigma_1_1}\\
        $\epsilon$F-satisfiability &
            $\Sigma^1_1$-hard		$\la \omega, 3, 2\ra$
                   \cite[thm~4.2]{val_is_pi_1_1}&
                $\Pi^0_1$-complete	$\la*\ra$
                    \cite[prop~3.2]{Terwijn_dec}\\
        \hline\hline
    \end{tabular}
    \caption{Summary of current knowledge on general $\epsilon$F-satisfiability
    and $\epsilon$F-validity}
    \label{tble:F_general_computability}
\end{table}
%


\subsection{Finite and Countable Concepts}

\begin{defn}
An $\epsilon$-model $\pmd$ is \textbf{finite} iff $|\M|$ is finite.
Similarly,
an $\epsilon$-model $\pmd$ is \textbf{countable} iff $|\M| \le \aleph_0$.
\end{defn}

Kuyper and Terwijin ahve shown well-formed model-theoretic properties exist for
E-logic (and by duality, F-logic). For example, a downward Lowenheim-Skolem
theorem always allows one to work on a model of continuum size
\cite[Thm.~4.6]{model_theory_of_e_logic}, and a variant of ultrapower construction
works for a weakened definition of $\epsilon$-models
\cite[Sec.~8]{model_theory_of_e_logic}. But they are futile in the finite setting
for the same reason that their classical counterparts do not work in classical
finite model theory.

Despite this difficulty, finite and countable $\epsilon$-models are still very
conducive to analysis because from them we automatically get (finite
and countable) probability models.

\newcommand{\pmdd}{\pmodel M {D'}}
\begin{lma}
Let $\pmd$ be an $\epsilon$-model. If $\D'$ extends $\D$, then $\pmdd
\equiv_\epsilon \pmd$. \label{lma:measure_extension_implies_equivalence}
\end{lma}
\begin{proof}
By induction on formula complexity, we show $\pmodel M {D'} \emodels \phi \iff
\pmd \emodels \phi$. All cases other than $\forall$ are trivial, as they don't
involve measures.

For $\phi(\vec x; \vec p) = \forall y \psi(y, \vec x; \vec p)$, we have
\begin{align*}
& \pmd \emodels \phixp\\
\iff& \D(a \in \M: \pmd \emodels \psi(a, \vec x; \vec p)) \ge 1 - \epsilon\\
\iff& \D'(a \in \M: \pmd \emodels \psi(a, \vec x; \vec p)) \ge 1 - \epsilon\\
\iff & \pmdd \emodels \phixp
\end{align*}
where the middle equivalence derives from the fact that $\D'$ agrees with $\D$
on $\dom D$.
\end{proof}

\begin{lma}[Tarski \cite{birkhoff}]
Every finitely additive measure $\D$ on a set $X$ can be extended to a finitely
additive measure $\D'$ so that $\dom \D' = \powerset(X)$.
\end{lma}

The above two lemmas allow us to ``complete" $\epsilon$-models in the following
sense:

\begin{prop}\label{finite_model_extension}
Let $\pmd$ be any finite or countable $\epsilon$-model. $\D$ can be extended to
a measure $\D'$ with $\dom \D' = \powerset(M)$. Therefore, $\pmodel M {D'}$ is a
probability model and by (\ref{lma:measure_extension_implies_equivalence}),
    $$\pmd \equiv_\epsilon \pmodel M {D'}.$$
\end{prop}

\begin{proof}
The case of finite $\epsilon$-models $\pmd$ follows directly from the
lemma since $\D$ is countably additive iff $\D$ is finitely additive.

For the case of countable models $\pmd$, notice that $\D$ cannot be atomless, or
else $M$ would have to be uncountable. Suppose $a_0 \subseteq M$ is an atom. If $M -
a_0$ is not null, then the restriction of $\D$ on $M - a_0$ by the same
reasoning must not be atomless, and so there is an atom $a_1 \subseteq M - a_0$.
By induction, $M$ can be expressed as the disjoint union of at most countable
number of atoms and a null set. Hence it suffices to show that (by assuming $M$
is an atom itself) $\D$ extends to $\powerset(M)$ when $\D$ is a 0-1 measure.

Suppose not. Then every measure defined on all of $\powerset(M)$ is inconsistent
with $\D$.
In particular, the measure $I_x$ concentrating measure 1 on an element $x \in M$
cannot be a extension of $\D$, and that can happen only if there is
$T_x \in \D$ with measure 0 but contains $x$. Thus there
is a countable set $\{T_x: x \in M\}$ satisfying $\D(T_x) = 0$ and $T_x \ni x$.
But by countable subadditivity $0 = \sum_x \D(T_x) \ge \D(M) = 1$, which is a
contradiction.
\end{proof}

Because every finite or countable $\epsilon$-model $\pmd$ can be taken to have
$\D$ everywhere defined on $\powerset(M)$, we treat $\D$ as a point function on
$M$. It then makes sense to speak of $\D(a)$ for $a \in M$ and in particular,
elements of measure zero, or \textit{null elements}.

Finally, we define the main objects of study in this paper.

\begin{defn}
For X = F or E:

A sentence $\phi$ is said to be {\bf finitely $\epsilon$X-valid} and
is called a {\it finite $\epsilon$X-validity} (resp. {\bf countably
$\epsilon$X-valid} and {\it countable $\epsilon$X-validity}) if $\phi$ is
$\epsilon$X-satisfied by all finite (resp. countable) $\epsilon$-models.

Likewise,
a sentence $\phi$ is said to be {\bf finitely $\epsilon$X-satisfiable} and is
called a {\it finite $\epsilon$X-satisfiable}
(resp. {\bf countably $\epsilon$X-satisfiable} and {\it countable
$\epsilon$X-satisfiable}) if $\phi$ is $\epsilon$X-satisfied by some finite
(resp. countable) $\epsilon$-models.
\label{def:finite_countable_validity_satisfiability}
\end{defn}

In addition, to distinguish between these concepts and the similar concepts in
FOL, we say that
a formula $\phi$ is {\it finitely classically valid} if $\phi$ is satisfied
by all finite first order models;
a formula $\phi$ is {\it finitely classically satisfiable} if $\phi$ is
satisfied by some finite first order model.

Note that by the Lowenheim-Skolem theorem, what would be the concept of
``countably classically valid'' coincides with unrestricted validity in first
order logic.
This is not true for $\epsilon$E- or $\epsilon$F-logic unconditionally.
Terwijn and Kuyper provide a counterexample in
\cite[exmp~4.5]{model_theory_of_e_logic}.

\subsection{Examples and Applications}
\label{ssec:examples_and_applications}

Here the goal is two-fold:
1) we clarify concepts developed in the previous sections through examples,
and 2) we also note possible applications of $\epsilon$E- and $\epsilon$F-logic,
in part to motivate the main theorems.
For the second point, we feel it is illuminating to mention results in later
sections.
Readers are encouraged to check that these anachronism are correctly applied
after perusing their respective expositions.

We first exhibit some examples that highlight the difference in semantics
between classical first order logic and our $\epsilon$E- and $\epsilon$F-logics.

\begin{exmp}
Let $\phi := \exists x \forall y [x = y]$ where $=$ is true equality.
Classically, a model $\M \models \phi$ iff $M$ is a singleton.
This also holds in $\epsilon$F-logic for all $\epsilon < 1$.
In $\epsilon$E-logic, $\pmd \emodels \phi$ iff there is a singleton subset
$\{a\} \subseteq M$ such that $\D(\{a\}) \ge 1 - \epsilon$.
\end{exmp}

\begin{exmp}
Let $\phi := \forall x \exists y [x = y]$ where $=$ is true equality.
Classically, $\phi$ holds in every nonempty model.
This is true also for $\epsilon$E-logic.
But $\pmd \fmodels \phi$ iff every element of $M$ has $\D$-measure greater than
$\epsilon$.
In particular, when $\epsilon = \f 1 n$, $M$ must have less than $n$ elements;
when $\epsilon = 0$, $M$ is at most countable.
\end{exmp}

\begin{exmp}
Let $\phi := \exists x \forall y [x \not= y]$ where $=$ is true equality.
$\phi$ is a contradiction in classical first order logic and in
$\epsilon$F-logic,
but $\pmd \emodels \phi$ iff there is a singleton subset $\{a \} \sbe M$ with
$\D$-measure less than $\epsilon$.
\end{exmp}

\begin{exmp}
Let $\phi := \forall x \exists y [x \not= y]$ where $=$ is true equality.
Classically and in $\epsilon$E-logic, $\M \models \phi$ iff $|M| \ge 2$.
In $\epsilon$F-logic, $\pmd \fmodels \phi$ iff every singleton subset of $M$ has
$\D$-measure at most $1 - \epsilon$.
\end{exmp}

\begin{exmp}
Let $\psi(x)$ be a formula with a single free variable.
In first order logic,
    $$\psi(x),\ \forall y [ x = y \limplies \psi(y)],\ \exists y [ x = y \land
    \psi(y)]$$
are equivalent formulas.

In $\epsilon$E-logic, for any parameter $a \in M$,
    $$\pmd \emodels \psi(a) \iff \pmd \emodels \exists y[y = a \land \psi(y)]$$
but $\pmd \emodels \forall y [ y = a \limplies \psi(y)]$ whenever
$\{b: b \not= a\} \sbe M$ has inner $\D$-measure at least $1 - \epsilon$.

Likewise, in $\epsilon$F-logic, for any parameter $a \in M$,
    $$\pmd \fmodels \psi(a) \iff \pmd \fmodels \forall y [ y = a \limplies \psi(y)]$$
but $\pmd \not\fmodels \exists y[y = a \land \psi(y)]$ whenever
$\{a\} \sbe M$ has $\D$-measure $\le \epsilon$.
\end{exmp}

In first order logic, a common way to assess whether a theorem $\psi$ follows
from a set of axioms $T$ is to pass the sentence $\tau := \neg(T \limplies
\psi)$ to some resolution algorithm R.
$T$ implies $\psi$ iff R resolves $\tau$ to a contradiction.

But the obvious analogue for $\epsilon$E-logic cannot work:
$\pmd \emodels \phi \limplies \psi$ is not equivalent to
    $$\pmd \emodels \phi \implies \pmd \emodels \psi$$
because of the paraconsistency of $\epsilon$E-logic.
However, from definition
(\ref{def:E-truth}), $\pmd \emodels \phi \limplies \psi$ iff
    $$\pmd \emodels \neg \phi \lor \psi$$
which is equivalent to
    $$\pmd \not\emodels \neg \phi \implies \pmd \emodels \psi.$$
Rephrasing using $\epsilon$F-logic then,

\begin{prop}[deduction theorem]
$\pmd \emodels \phi \limplies \psi$ iff
    $$\pmd \fmodels \phi \implies \pmd \emodels \psi.$$
By duality, $\pmd \fmodels \phi \limplies \psi$ iff
    $$\pmd \emodels \phi \implies \pmd \fmodels \psi.$$
\end{prop}

Thus, axioms $T$, interpreted inside $\epsilon$F-logic, (metalogically) imply
theorem $\psi$, interpreted inside $\epsilon$E-logic, iff
    $$T \limplies \psi$$
is an $\epsilon$E-validity.

By results (\ref{sat_reduction''}) and (\ref{fsat_reduction''}) of Kuyper and
Terwijn that we record in later sections,
we can even interpret each quantifier in $T$ (resp. in $\psi$) in
$\alpha$F-logics (resp. $\alpha$E-logics) for different $\alpha$s, in this
sense:

\newcommand{\noteq}{$\nabla$}
\newcommand{\MyQuote}[1]{\vspace{\baselineskip}
     \parbox{.8\linewidth}{#1}\hspace*{.1\linewidth}(\noteq)\\
     [\baselineskip]}
\MyQuote{
    Let $T$ be in prenex normal form $\varqs n\phi(\vec x)$.
    For every $\varq_i = \exists$, we translate $\varq_i x_i$ as
        ``there exists a set $A_i$ with measure $> \alpha_i$ such that
        for each $x_i \in A_i$, \ldots;''
    The quantifier $\forall$ is interpreted as usual.
    Each $\alpha_i$ can be any arbitrary rational number in $[0, 1]$,
    independent of what other $\alpha_j$s are.

    Similarly,
    let $\psi$ be in prenex normal form $\varqs m\phi'(\vec x)$.
    For every $\varq_i = \forall$, we translate $\varq_i x_i$ as
        ``there exists a set $B_i$ with measure $\ge 1 - \beta_i$ such that
        for each $x_i \in B_i$, \ldots.''
    The quantifier $\exists$ is interpreted as usual.
    Each $\beta_i$ can be any arbitrary rational number in $[0, 1]$,
    independent of what other $\beta_j$s are.

    The implication $T \implies \psi$ (resp. $\psi \implies T$) under
    these translations can be expressed as some sentence in $\epsilon$E-logic
    (resp. $\epsilon$F-logic).
}

As a corollary,
if $T$ is a conjunction of sentences $\{\Lambda_j\}_{j=1}^l$ and
$\psi$ is a disjunction of sentences $\{\Gamma_i\}_{i=1}^k$, then
for any finite $\{\alpha_j\}_{j=1}^l, \{\beta_i\}_{i=1}^k$ of
rational numbers in $[0, 1]$,
there is some sentence $\Phi$ such that the following are equivalent:
\begin{itemize}
    \item for all finite $\pmd$, if $\pmd$ simultaneously
    $\alpha_j$F-satisfies each $\Lambda_j$, then
        $$\pmd \emodelsp{\beta_i} \Gamma_i\qquad
            \text{for some $i \in \indset k$}.$$
    \item $\Phi$ is finitely $\epsilon$E-valid.
\end{itemize}

Though such a form of deduction theorem may  not seem useful at first, it
nevertheless allows the expression of statements regarding many types of
mathematical objects, including concept classes
in CLT, graphs with weighted vertices,
graphs with weighted edges, and artificial neural networks.

In what follows, we give examples of such expressions and their
meanings in both $\epsilon$E-logic and $\epsilon$F-logic.
For ease of reading, we will write
$\foralleps$ for $\forall$ in the context of $\epsilon$E-logic and
$\existseps$ for $\exists$ in the context of $\epsilon$F-logic.
As noted, in a single sentence, the $\epsilon$ parameter may vary, so
it is meaningful to write $\existseps[0] x \existseps[1/2] y
\existseps[\epsilon] z.$
Each boolean symbol whose meaning has not been defined is a shorthand for
the usual composition of symbols $\neg, \land, \lor.$

These examples will hopefully give a rough picture of the implications of our
decidability and undecidability results.

\begin{exmp}[PAC learning]
\newcommand{\cC}{\mathcal{C}}
\newcommand{\oracle}{\mathrm{EX}}
Define the following:
\begin{itemize}
\item an \textit{example space} is a probability space $(X, \D)$
with probability distribution $\D$;
\item a \textit{concept class} $\cC$ is a collection $\{U \sbe X\}$ of subsets
of $X$, identified with their indicator functions;
\item for $U \in \cC$, an example oracle $\oracle(X, \D, U)$ is a device that randomly
returns a pair $(x \in X, U(x))$ for every invocation.
The pair is sampled according to $\D$.
\end{itemize}
In the basic PAC learning model \cite{valiant_PAC}, we are given $X$, $\cC$, and
an example oracle $\oracle$ that emits elements of $X$ according to an unknown
distribution $\D$ and unknown concept $U \in \cC$.
We wish to efficiently find a concept close enough to $U$,
in the following sense:
\begin{quote}
We have a probabilistic algorithm that,
for all error parameters $\epsilon$ and $\delta$,
for all distributions $\D$ on $X$,
in time polynomial in $\f 1 \epsilon$ and $\f 1
\delta$, returns $U' \in \cC$ with
    $$\Pr_{x \sim \D}[U'(x) \not= U(x)] < \epsilon$$
more than $1 - \delta$ of the time.
Such a $U'$ is called an \textit{approximation concept}.
\end{quote}

\newcommand{\bclass}{\mathfrak{B}}
Typically, $X$ is taken to be $\bclass_s = \{0, 1\}^s$ for some $s$, and $\cC$
is a subclass of the boolean functions on $\bclass_n$.

\newcommand{\clsf}{\mathfrak{c}}
It is important to note that this learning model assumes
\begin{equation}
\text{for the given $\oracle$, some approximation concept $U'$ exists in $\cC$.}
\tag{$\wp$}\label{stmt:PAC_assume}
\end{equation}
In particular, this happens if $\oracle$ labels elements according to some
concept in $\cC$.
But in practice this may not always be the case \footnote{hence the research
into \textit{agnostic learning}; see \cite{kearns_agnostic} for an overview}.
In $\epsilon$E-logic with a language having monadic predicate $P_1, \ldots,
P_s$ and $\clsf$, we can express this assumption.

In what follows,
for a general vector $w$, $w_i$ will represent the $i$th value of $w$;
$\M$ is a first order model over universe $\bclass_s$, and
    $$P^\M_i(v) = v_i,\qquad
        \clsf^\M(v) = \text{the label generated by $\oracle$};$$
$\D$ is a probability measure everywhere defined over $M$.
\renewcommand{\forall}{\foralleps}
\begin{enumerate}
    \item (Point class). The point concept class over $\bclass_s$ is the
    collection
        $$\{\{v\}: v \in \bclass_s\}.$$
    In other words, each concept labels exactly one point of $\bclass_s$ as
    1 and the rest as 0.
    Assumption (\ref{stmt:PAC_assume}) holds for the point class iff
        $$\pmd \emodels \exists x \forall y [\clsf(y) \lequiv
            \bigwedge_{i=1}^s(P_i(x) \lequiv P_i(y))].$$
    (Of course it is much more convenient to use equality, but we refrain in order
    to take advantage of the decidability of monadic relational languages)
    \item (Parity class). For each $v \in \bclass_s$, let
        $$O_v(x) := v \cdot w \mod 2$$
    where $(-) \cdot(-) \mod 2$ is the dot product in the vector space
    $(\Z_2)^s$.

    The parity concept class over $\bclass_s$ is the collection
        $$\{O_v: v \in \bclass_s\}.$$

    Then assumption (\ref{stmt:PAC_assume}) holds for the parity class iff
        $$\pmd \emodels \exists x \forall y [\clsf(y) \lequiv
            \bigoplus_{i=1}^s (P_i(x) \wedge P_i(y))].$$
    \item (Conjunction class). The conjunction concept $C_w$ represented by a
    vector $w \in \{-1, 0, 1\}^s$ labels $v \in \bclass_s$ as 1 iff
    for each $i \in \indset s$ such that $w_i \not = 0$, $v_i = (w_i + 1)/2$.
    The conjunction concept class is the collection of $C_w$ over all $w \in
    \{-1, 0, 1\}^s$.

    Assumption (\ref{stmt:PAC_assume}) holds for the conjunction class iff
        $$\pmd \emodels
                \exists x \exists y \forall z \lp\clsf(z) \lequiv
                    \bigwedge_{i=1}^s[
                    (P_i(x) \land P_i(y) \limplies P_i(z)) \land
                    (\neg P_i(x) \land \neg P_i(z) \limplies \neg P_i(z))]\rp$$
    Here we code each conjunction concept $C_w$
    using 2 bitstrings $x$ and $y$:
    if $x_i = y_i = 1$, $w_i = 1$;
    if $x_i = y_i = 0$, $w_i = -1$;
    otherwise $w_i = 0$.
    \newcommand{\DL}{\mathrm{DL}}
    \item (1-Decision lists). Let $Z$ be a triple $(\alpha, \beta, b)$
    where $\alpha, \beta \in \bclass_s$ and $b \in \{0, 1\}$.
    A 1-decision list $\DL_Z$ represented by $Z$ is the decision procedure that,
    on input $v \in \bclass_s$, runs as follows
    \begin{algorithmic}
        \If{$v_1 = \alpha_1$} output $\beta_1$
        \ElsIf{$v_2 = \alpha_2$} output $\beta_2$
        \State $\vdots$
        \ElsIf{$v_s = \alpha_s$} output $\beta_s$
        \Else \ output $b$
        \EndIf.
    \end{algorithmic}

    Let
    \begin{align*}
    \psi_i(x, y, w) &:= [P_1(x) \oplus P_1(w)] \land [P_2(x) \oplus P_2(w)]
    \land \cdots \\
    & \qquad \land [P_{i-1}(x) \oplus P_{i-1}(w)] \land [P_i(x) \lequiv P_i(w)]\\
    & \qquad \limplies [\clsf(w) \lequiv P_i(y)].
    \end{align*}

    $\psi_i$ represents a computation of $\DL_{(x, y, b)}$ (for any $b$) on
    input $w$ that proceeds to the $i$th \textbf{if} statement before returning.

    Let
    \begin{align*}
    \phi(x, y, z, w) &:= [P_1(x) \oplus P_1(w)] \land [P_2(x) \oplus P_2(w)]
    \land \cdots \\
    & \qquad \land [P_{s-1}(x) \oplus P_{s-1}(w)] \land [P_s(x) \oplus P_s(w)]\\
    & \qquad \limplies [\clsf(w) \lequiv P_1(z)]
    \end{align*}

    $\phi$ represents a computation of $\DL_{(x, y, z_1)}$ on input $w$ that
    proceeds to the \textbf{else} statement.

    Then assumption (\ref{stmt:PAC_assume}) holds for the 1-decision lists iff
         $$\pmd \emodels \exists x \exists y \exists z \forall w
            \lp\phi(x, y, z, w) \land \bigland_{i=1}^s \psi_i(x, y, w)\rp.$$

\end{enumerate}
\renewcommand{\forall}{\oldforall}

\textit{VC dimension} is an important quantity of concept
classes studied in CLT.
We will not define it here (consult \cite{kearns_CLT}) but wish to mention
to the computational learning theorists that, as these examples illustrate, when
a concept class $\cC$ over $\bclass_s$ has VC dimension $f(s)$,
assumption (\ref{stmt:PAC_assume}) can be expressed over $\pmd$ by a formula
with $O(f(s))$ number of quantifiers (exercise!).

On the other hand, in $\epsilon$F-logic, we can express certain interesting
conditions on the probability space $(M, \D)$.
\renewcommand{\exists}{\existseps}
\newcommand{\uniform}{\textsc{uniform}}
\newcommand{\atteff}{\textsc{atteff}}
\begin{enumerate}
    \item ``more than half the time, the label is 1":
        $$\pmd \fmodelsp{\f 1 2} \existseps[1/2] x \clsf(x).$$
    \item ``the probability that the $i$th bit is 1 is strictly between
    $\epsilon$ and $1 - \epsilon$'':
        $$\uniform^i_{|\f 1 2 - \epsilon|} := \pmd \fmodels \forall x \exists y
        [P_i(x) \oplus P_i(y)].$$
    As a straightforward generalization, the quantifier-free part of the
    expression can be swapped out for a more complicated boolean combination.
    \item ``most bits are irrelevant to the concept ---
        the label only depends on 2 fixed bits":
        $$\atteff_2 := \pmd \fmodels \biglor_{i < j} \forall x \forall y [
                (P_i(x) \lequiv P_i(y)) \land (P_j(x) \lequiv P_j(y))
                \limplies
                (\clsf(x) \lequiv \clsf(y))].$$

        This is the setting for \textit{attribute efficient learning}
        \cite{attribute_efficient_klivans}.
        \label{exmpitem:attribute_efficient}
\end{enumerate}
\renewcommand{\exists}{\oldexists}
As noted before, each $\epsilon$ in these examples can vary over the rationals
at will.
Imagine we want to find if $\uniform^i_{\delta_1}$ and $\uniform^j_{\delta_2}$
along with $\atteff_2$ would imply (\ref{stmt:PAC_assume}) for some concept
class $\cC$.
We may do so by querying for the $\epsilon$E-validity of a two-part sentence
    $$T \limplies \phi$$
derived from applying note (\noteq).

By the decidability of monadic relational languages
(\ref{thm:E-monadic_decidable} and \ref{thm:F-monadic_decidable}) established
later, these kinds of questions are all decidable, as long as all statements use
only unary predicates.

It is also possible to add a BIT relation to $\epsilon$E- and $\epsilon$F-logics
along the lines of the corresponding relation in classical finite model theory
\cite{libkin}.
One can then express concept classes over all size parameter $s$ with
one single sentence.
However, the complexity of deduction then becomes unknown.

\end{exmp}
\begin{exmp}[graphs with weighted vertices]
Let $\lL$ be a language with a single binary relation $E(\cdot, \cdot)$,
standing for the edge relation between two vertices. As in first order logic,
any $\epsilon$-model of this language is automatically a directed graph with at
most one edge between every pair of vertices (including loops). Moreover, in
such $\epsilon$-models, to each vertex of the graph structure is assigned a
weight in $[0, 1]$ such that the total sum of all vertex weights is 1; however,
the edges are not weighted.
Such graphs with weighted vertices can be used to model, among many things,
cities with populations and in general the PageRank algorithm.

Because in
$\epsilon$F-logic, $\forall$ is interpreted classically, we can express quite
a few properties of graphs:
\renewcommand{\exists}{\existseps}
\begin{enumerate}
    \item loopless:
        $$\forall x \neg E(x, x).$$
    \item undirected:
        $$\forall x \forall y [E(x, y) \lequiv E(y, x)].$$
    \item complete:
        $$\forall x \forall y [E(x, y) \land E(y, x)].$$
    \item bipartite with a fixed partition (if $\lL$ has a unary predicate $A$):
        $$\forall x \forall y [(A(x) \lequiv A(y)) \limplies \neg E(x, y) \land
                                                        \neg E(y, x)].$$
        Similarly, $k$-coloring can be expressed as well.
    \item in a simple graph, ``every 1-neighborhood collectively has weight more
    than $\epsilon$'' (if $\lL$ has equality):
        $$\forall x \exists y [y = x \lor E(x, y)].$$
    \item ``every directed triangle has collective measure more than $\epsilon$"
    (if $\lL$ has equality):
        $$\forall x \forall y \forall z [
            E(x, y) \land E(y, z) \land E(z, x) \limplies
                \exists w (w = x \lor w = y \lor w = z)].$$
    \item ``every element has a unique successor of positive measure" (when
    $\lL$ has equality):
        $$\forall x \existseps[0] y [E(x, y) \land
            \forall z (E(x, z) \limplies z = y)].$$
    \item ``there exists a set of `initial' vertices $A$ collectively with
    weight more than $\epsilon$ such that every $v \in A$ is connected to every
    vertex in the entire graph":
        $$\exists x \forall y E(x, y).$$
\end{enumerate}

\renewcommand{\exists}{\oldexists}
On the other hand, in $\epsilon$E-logic, we can express the likes of
the following.
\renewcommand{\forall}{\foralleps}
\newcommand{\isom}{\textsc{iso}}
\begin{enumerate}
    \item ``there is a clique of size $k$'':
        $$\exists \vec x \lp \bigland_{1 \le i < j \le k} E(x_i, x_j) \land
                                        E(x_j, x_i)\rp.$$
    In general, for any fixed graph $G$, we can express the existence of a
    subgraph isomorphic to $G$.
    \item ``there is a subgraph isomorphic to $G$ that carries weight
    $\ge 1 - \epsilon$":
        $$\exists \vec x \left[\isom_G(\vec x) \land
                \forall y \biglor_{i=1}^k y = x_i\right],$$
    where $\isom_G$ is a formula expressing $\vec x$ is isomorphic to $G$.
    \item in a simple graph, ``there is a single vertex $v$ that is connected to
    at least $1-\epsilon$ (by weight) of other vertices":
        $$\exists v \forall x E(x, v).$$
\end{enumerate}
\renewcommand{\forall}{\oldforall}

In $\epsilon$E-logic, we can also make weakened version of universal statements
from the $\epsilon$F-examples by replacing $\forall$ with $\foralleps[0] = \formone$.
The transformed sentences will then quantify over all elements of positive
measure, rather than all elements.
For instance, the sentences given for the loopless, undirected, complete, and
bipartite properties all carry over to apply to the subgraph consisting of all
nonnull vertices.

This weakened quantifier suffices in most cases.
In fact, later on, our proof of the undecidability of
finite $\epsilon$E-satisfiability (\ref{thm:Esat_Sigma1-hard}) depends heavily
on this competency of $\formone$.

In some cases, when $\epsilon > 0$, the quantifier $\existseps$ can also be
replaced (per note (\noteq))
with $\foralle$ without affecting the intended semantics very much.
For instance, the property ``every 1-neighborhood collectively has weight more
than $\epsilon$'' differs very little from ``every 1-neighborhood collectively
has weight at least $\epsilon$'' in most imaginable applications.

\end{exmp}

\begin{exmp}[graphs with weighted edges]
\label{exmp:graph_weighted_edges}
Instead of assigning a measure to vertices, often we want to assign numbers to
edges of a graph, for example in a MAX-FLOW or a path-finding problem.

Let $\lL$ be a language with binary relations $I(\cdot, \cdot)$,
$C(\cdot, \cdot)$, and $D(\cdot, \cdot)$.
Here $I(x, y)$ represents that the codomain of edge $x$ equals the domain of
edge $y$; $D(x, y)$ (resp. $C(x, y)$) represent that edges $x$ and $y$ have the
same domain (resp. codomain).
Thus, any graph with weighted edges $\{e_i\}_{i \in M}$ such that the total
weight equals 1 is automatically a probability model in $\lL$.

Conversely, suppose $\pmd$ has an everywhere defined measure
$\D$ and classically satisfies the axioms of

\begin{itemize}
\item ``$C$ and $D$ are equivalence relations'':
\begin{equation*}
\begin{aligned}
    \forall x & C(x, x)\\
    \forall x \forall y & C(x, y) \lequiv C(y, x) \\
    \forall x \forall y \forall z & C(x, y) \land C(y, z) \limplies C(z, x) \\
\end{aligned}
                \tag{EQR}	\label{axiom:equiv_do_co}
\end{equation*}
along with the analogues for $D$.
\item ``incidence relation respects domain and codomain":
\begin{equation}
\begin{aligned}
    \forall x \forall y & (C(x, y) \limplies
                    \forall z [I(x, z) \lequiv I(y, z)])\\
    \forall x \forall y & (D(x, y) \limplies
                    \forall z [I(z, x) \lequiv I(z, y)]).\\
\end{aligned}
                \tag{IDC}	\label{axiom:inc_do_co}
\end{equation}
\item ``domain and codomain respect incidence relation":
\begin{equation}
\begin{aligned}
    \forall x \forall y & (I(x, y) \limplies
        \forall z [C(z, x) \lequiv I(z, y)])\\
    \forall x \forall y & (I(x, y) \limplies
        \forall z' [D(z', y) \lequiv I(x, z')])
\end{aligned}
                \tag{DCI}	\label{axiom:do_co_inc}
\end{equation}
\item ``domain and codomain are unique'':
\begin{equation*}
\begin{aligned}
    \forall x \forall y \forall z & [
        I(x, y) \land I(x, z) \limplies D(y, z)] \\
    \forall x \forall y \forall z & [
        I(y, x) \land I(z, x) \limplies C(y, z)] \\
\end{aligned}
                \tag{!DC}	\label{axiom:unique_do_co}
\end{equation*}
\end{itemize}

Let $V_D$ and $V_C$ be respectively the equivalence classes modulo $D^\M$ and
$C^\M$.
For every element $a$, denote its equivalence class in $V_D$ by $\llb a\rrb_D$
and that in $V_C$ by $\llb a \rrb_C$.
By axiom (\ref{axiom:do_co_inc}), $\llb a \rrb_D \in V_D$ can be identified
with $\llb b \rrb_C \in V_C$,
$\llb a\rrb_D\sim \llb b \rrb_C$, if $I^\M(b, a)$.
Thus we can form the vertex set $V = (V_D \cup V_C)/\sim$.
Axiom (\ref{axiom:inc_do_co}) implies that the relation $I^\M$ induces
a relation $I^\M_C$ on $V_C \times M$.
By factoring through the identification $\sim$, the relation $I^\M_C$ can be
treated as a relation on $V \times M$ with the property that, for every element
$a$ of $M$, $I^M_C(\llb a\rrb_D, a)$.
But axiom (\ref{axiom:unique_do_co}) says that this $I^\M_C$ is in
fact a function $M \to V$.
We therefore retrieve the domain function
$\dom: M \to V, a \mapsto \llb a \rrb_D$.
By the same reasoning, we also derive the codomain function
$\cod: M \to V, a \mapsto \llb a \rrb_C$.
These data then uniquely determine a graph with edges $M$ and vertices $V$.

As $\epsilon$F-logic interprets $\forall$ classically, it can convey
the axioms along with many of the usual properties of graphs.
\newcommand{\pathmin}[2][\epsilon]{\textsc{pathmin}^{#2}_{#1}}
\renewcommand{\exists}{\existseps}
\begin{enumerate}
    \item loopless:
        $$\forall x \neg I(x, x).$$
    \item ``no more than one edge per pair of vertices" (if $\lL$ has equality):
        $$\forall x \forall y [D(x, y) \land C(x, y) \limplies x = y].$$
    \item bidirectional: ``each edge $x$ has a corresponding edge with positive
        weight that goes in the opposite direction":
        $$\forall x \existseps[0] y [I(x, z) \land I(z, x)].$$
    On the other hand, we cannot express \textit{un}directedness without
    changing the axioms.
    The reader is encouraged to work out the axioms for a simple graph with
    weighted edges.
    \item complete: ``for any two edges $x$ and $y$, there is an edge with
        positive weight that connects $x$ to $y$":
        $$\forall x \forall y \existseps[0] z [I(x, z) \land I(z, y)].$$
    \item ``every directed triangle has collective weight $>\epsilon$''
        (if $\lL$ has equality):
        $$\forall x \forall y \forall z [I(x, y) \land I(y, z) \land I(z, x)
                            \limplies
                            \exists w (w = x \lor w = y \lor w = z)].$$
       For any fixed $k$, we can also make the analogous statement for
       $k$-cycles.
    \item ``every length $k$ path from $x$ to $y$ has weight $> \epsilon$"
        (if $\lL$ has equality):
        $$\pathmin k (x, y) := \forall \vec x \lp
                        \bigland_{i=1}^{k-1} I(x_i, x_{i+1}) \limplies
                        \exists y \lp \biglor_{i=1}^k y = x_i \rp\rp.$$
        For any fixed graph $G$ of size $k$,
        the constructions in this and the last items generalize to make
        statements of the form
        ``any subgraph isomorphic to $G$ has total weight $>\epsilon$.''

\end{enumerate}

\renewcommand{\exists}{\oldexists}

As with example (\ref{exmp:graph_weighted_edges}), in $\epsilon$E-logic, for any
fixed graph $G$, we can express the existence of a subgraph isomorphic to $G$.
We can also assert that some such subgraph has weight $\ge 1 - \epsilon$.
These properties may be desirable when working with MAX-FLOW problems.

Finally, we can also transform statements in $\epsilon$F-logic into
weaker statements in $\epsilon$E-logic by replacing $\forall$ with $\formone$
and $\existseps$ with $\foralle$.
With emphasis, we note that all axioms
(\ref{axiom:equiv_do_co}), (\ref{axiom:do_co_inc}), (\ref{axiom:inc_do_co}), and
(\ref{axiom:unique_do_co}) of graphs with weighted edges are universal
sentences.
Therefore, as long as the presence of zero-weight edges present no difficulty,
we can also investigate implications
    $$T \implies \psi$$
with $\psi$ interpreted in $\epsilon$F-logic.
\end{exmp}

\newcommand{\activated}{\textsc{actv}}
\begin{exmp}[artificial neural networks]

\newcommand{\weight}{\mathfrak{w}}
\textit{Artificial neural network (ANN)} is a very popular biologically
inspired technique in machine learning that is often used in pattern recognition
\cite{mackay}\cite{russel_norvig}.
Each ANN is a directed graph in which each edge $e$ has weight $\weight(e)$.
Its nodes are called \textit{neurons} and its edges are called
\textit{connections}.
If neuron $\eta$ connects to neuron $\zeta$ via connection $e$, we say $\eta$
\textit{feeds into} $\zeta$ via $e$ (written $\eta \xrar e \zeta$), $\eta$ is
the \textit{presynaptic neuron} of $e$, and $\zeta$ is the \textit{postsynaptic
neuron} of $e$.
Each neuron is either \textit{activated} or not.
Its state at time $t+1$ depends on the activation states at time $t$ of the
neurons that feed into it.
The exact update rule may vary in different neural networks, but usually
it is implemented as a \textit{linear threshold function}:
\newcommand{\threshold}{\mathfrak{T}}
\begin{quote}
Each neuron $\eta$ has a \textit{threshold value} $\threshold$ such that
$\eta$ is activated at time $t+1$ iff
    $$\sum_{\zeta \xrar e \eta}
            \weight(e) \cdot \|``\zeta \text{ activated at time } t"\| >
            \threshold.$$
\end{quote}

Like in the previous example, ANNs can be represented by finite
$\epsilon$-models $\pmd$ with $\dom D = \powerset(M)$ of the language $\lL$ with
binary relations $I$, $D$, and $C$.
The measure $\D(a)$ of each element $a$ of $M$ correspond to the weight
$\weight(a)$ in the ANN.
If the threshold $\threshold$ is fixed across all neurons, then
the linear threshold update rule can be expressed in $\epsilon$F-logic.

We introduce new predicates $\activated_t(x)$ that represents whether
the presynaptic neuron of edge $x$ is activated at time $t$.
It satisfies the following relations for each $t$.

\begin{enumerate}
    \item ``Suppose $x$ and $y$ have the same presynaptic neuron.
    Then $\activated_t(x)$ holds iff $\activated_t(y)$ holds'':
        $$\forall x \forall y (D(x, y)
                        \limplies
                        [\activated_t(x) \lequiv \activated_t(y)]).$$
    \item the linear threshold update rule:
    $$\forall x (\activated_{t+1}(x) \lequiv
        \existseps[\threshold] y [I(y, x) \land \activated_t(y)]).$$
\end{enumerate}

\newcommand{\inp}{\mathfrak{I}}
\newcommand{\outp}{\mathfrak{O}}
In a typical usage of ANN, there are two sets of distinguished neurons
$\inp$ and $\outp$ called \textit{input neurons} and \textit{output neurons}.
At the beginning, each neuron of $\inp$ is activated or deactivated according
to an input bitstring, for example derived from a digital image.
All other neurons are not activated.
After some time $t$, the activation states of the neurons of $\outp$ are
returned as a bitstring.
Continuing our example, we might desire the output of 1 from every output neuron
iff the image is of a butterfly.

Imagine we are interested in whether some property $\Phi$ of ANN implies some
property $\Psi$.
If we can phrase $\Phi$ as a sentence to be interpreted under $\epsilon$F-logic
and $\Psi$ as a sentence to be interpreted under $\epsilon$E-logic,
then we can answer this question by querying for the $\epsilon$E-validity of
    $$\Phi \land \Lambda \limplies \Psi,$$
where $\Lambda$ is the conjunction of the axioms of the graph from
example (\ref{exmp:graph_weighted_edges}) and the axioms of $\activated_t$ from
this example.
\end{exmp}

Unlike the example of PAC learning, we cannot say with certainty whether any or
all of the theories of finite graphs with weighted vertices, finite graphs with
weighted edges, or finite artificial neural networks are decidable.
The main theorems of this paper will establish that the naive method is out of
the picture: there is no general deduction mechanism for $\epsilon$E- or
$\epsilon$F-logics when restricted to finite $\epsilon$-models.
In particular, this result holds even when restricted to first order languages
with a finite number of binary relations and an infinite number of unary predicates.
But that is not enough to determine the exact computability of the above
theories, each of which uses only a finite number of unary predicates.
(Even for ANN, in almost all use cases, only a finite number of $\activated_t$
predicates are considered).
Weakening the language requirement remains a major research area in $\epsilon$E-
and $\epsilon$F-logics.

Related to the issue of decidability is expressability.
We note in passing that despite these examples, $\epsilon$E- and
$\epsilon$F-logics still have nontrivial limitations in expression power.
These limitations derive in many cases from the limitations of
first order logic itself.
A detailed discussion of impossibility results in expressability is outside
the scope of this paper, but we mention that quite a few techniques for
first order logic, like locality, carry over to our probability logics.
The interested reader is advised to consult \cite{libkin}.

%

\section{Validities and Satisfiabilities}

\label{sec:finite_val_sat_monadic}

\subsection{Finite and Countable 0E-Satisfiabilities}
\label{ssec:0Fval}
\newcommand{\vecy}[1][y]{\vec{\underline{#1}}}

In contrast to first order logic, where Trachtenbrot's theorem implies
that an effective calculus for deducing true theorems over finite models cannot
exist, we show here that finite and countable 0F-validities are decidable.
Moreover, the results of this subsection apply to any first order language.
Consequently, finite and countable 0E-satisfiability are also decidable
regardless of language.

(Recall that $\vec x$ is a shorthand for a sequence of variables $x_1, x_2,
\ldots, x_n$ for some $n \ge 0$, and $\forall \vec x$ is a shorthand for
$\forall x_1 \forall x_2 \cdots \forall x_n$).
\begin{lma}[validity conversion]
Let $\varq_i \in \{\forall, \exists\}$ represent quantifiers, and
$$\phi:= \varq_1 x_1 \varq_2 x_2 \cdots \varq_n x_n \psi(\vec x)$$
where $\psi$ is a quantifier free formula.
Suppose $I = \{i_1, i_2, \ldots, i_k\} \sbe \{1, \ldots, n\}$ is an enumeration
of all indices $i$ such that $\varq_i = \forall$.
Define
$$\phi^*(y):= \forall x_{i_1} \forall x_{i_2}\cdots \forall
x_{i_k} \psi(\vecy, x_{i_1}, \vecy, x_{i_2}, \vecy,
\ldots, \vecy, x_{i_k}, \vecy),$$
where each $\vecy$ denote a block $y, y, \ldots, y$ of $y$ repeated some number
of times, depending on the location of $\vecy$. In other words, in $\phi^*$, all
$x_j$ with $j \not \in I$ has been substituted with the free variable $y$.

Let
$$\phi' := \forall y \phi^*(y).$$

Then $\phi$ is finitely 0F-valid iff $\phi'$ is finitely classically valid.
$\phi$ is countably 0F-valid iff $\phi'$ is classically valid.
\end{lma}

\begin{proof}
\newcommand{\eik}{\mathcal E_{i, k}}
\newcommand{\eikp}[2]{\mathcal E_{#1, #2}}
\newcommand{\vk}{\mathcal V_k}
\newcommand{\vkp}[1]{\mathcal V_{#1}}
\newcommand{\pve}{(\vk, \eik)}
\newcommand{\pvep}[2]{(\vkp{#1}, \pve{#1}{#2})}
Let's consider the finite validity claim of the theorem. The countable validity
portion is almost exactly the same.

(\textit{$\phi$ finitely 0F-valid $\implies$ $\phi'$ finitely classically
valid})
Let $\mathcal V_k$ be a classical model of size $k$, with
universe $\{1, 2, \ldots, k\}$. Define measures $\mathcal E_{i, k}$ on it such
that $\eik(i) = 1$ and $\eik(j) = 0,\ \forall j \not= i$.
If $\pve \fmodelsp 0 \phi$, then all the $x_j, j \not \in I$ (i.e. all those
with an existential quantifier) must be interpreted as $i$ since $i$ has measure
1.
Hence $\pve \fmodelsp 0 \phi$ implies $\vk \models \phi^*(i)$.

Since $\phi$ is finitely 0F-valid, for any fixed $k$, $\pve \fmodelsp 0 \phi$
and thus $\vk \models \phi^*(i)$ hold for all $1 \le i \le k$. Therefore,
$$\vk \models \forall y \phi^*(y) \implies \vk \models \phi'$$
Now vary $k$, and we conclude that $\phi'$ is classically finitely valid.

Note that the above reasoning did not use finiteness in an essential way. In
fact, slightly modifying the argument shows that $\phi$ 0F-valid for all
0-models of size $\kappa$ implies $\phi'$ classically valid for all first order
models of size $\kappa$.

(\textit{$\phi'$ finitely classically valid $\implies$ $\phi$ finitely
0F-valid})
Let $\pmd$ be a 0-model with the universe $M = \{1, \ldots, k\}$.
By proposition (\ref{finite_model_extension}), we can take $\D$ to be defined on
all subsets of $M$. Since $\phi'$ is satisfied by all classical models, $\mathcal M
\models \phi' \implies \mathcal M \models \forall y \phi^*(y)$.
Because $\M$ is finite (this is the only place where the finiteness is used;
substitute countability for the countable case), there must be an element $a \in
M$ with positive measure.
Then $\M \models \phi^*(a)$, meaning that all the $\exists$ bindings (with the
interpretation under $\fmodels$ of measure strictly positive) in $\phi$ are
realized by $a$.
Thus $\pmd \fmodelsp 0 \phi$.
Since $\pmd$ is an arbitrary finite 0-model, $\phi$ is finitely 0F-valid.

\end{proof}

\begin{lma}
Let
$$\phi(\vec y):=\forall \vec x \psi(\vec x, \vec y)$$
be a universal formula, where $\psi$ is quantifier free.

The following are equivalent:
\begin{enumerate}
    \item $\phi$ is finitely classically valid. \label{cfinval}
    \item $\phi$ is classically valid. \label{cval}
\end{enumerate}
\label{0ftautologycollapse}
\end{lma}

\begin{proof}
Certainly, $(\ref{cval}) \implies (\ref{cfinval})$.
It suffices to show $(\ref{cfinval}) \implies (\ref{cval})$.

Because for any first order model $\M$, $\M \models \phi(\vec y) \iff
\M \models \forall \vec y \phi(\vec y)$, we assume that $\phi$ is a sentence
$\forall \vec x \psi(\vec x)$. Let $n = |\vec x|$.

Suppose $\phi$ is satisfied by all finite (classical) models but there is an
infinite model $\M \models \neg \phi$.
Then there is a tuple $\vec a \in M^n$ such that $\mathcal M \models
\neg \psi(\vec a)$.
We form a finite model $\M'$ containing $\{a_i\}_{i=1}^n$ such that $\M' \models
\neg \psi(\vec a)$, which would yield a contradiction.

Let
$F_0 := \{a_i\}_{i=1}^n \cup
\{c^\M: c\text{ is a constant symbol that appears in $\phi$}\}$.
Given $F_i$, set
$$F_{i+1} := \{f^\M(\xi) : \xi \in F_i, f\text{ is a
function symbol that appears in $\phi$}\}$$

Then we define the universe
of our model to be $$M' = \bigcup_{i=0}^k F_k \cup \{r\}$$ where $k$ is
maximal number of times any function symbol appears in $\psi$, and $r$ is an
arbitrary new element.
$F_0$ is obviously finite, and given $F_i$ is finite, $|F_{i+1}| \le |F_i|
\cdot |\mathrm{length}(\phi)|$ is finite.
Thus each $F_i$ is finite and so $M'$ is finite.

The relations in $\M'$ will be the relations of $\M$ restricted to $M'$. For
each function symbol $f$ in the language, define
$$f^{\M'}(\xi) = \begin{cases}
f^\M(\xi) & \text{if $\xi \in F_i$ for some $i < k$}\\
r & \text{otherwise}
\end{cases}$$
if $f$ appears in $\psi$, and otherwise arbitrary.
Finally, for each constant symbol $c$, the interpretation is
$$c^{\M'} = \begin{cases}
c^\M & \text{if $c^\M \in H$} \\
r & \text{otherwise}
\end{cases}$$

It's easy to check that if $t$ is a term that appears in $\psi$, then $t^\M
= t^{\M'}$, and if $R$ is an $n$-ary relation that appears in $\psi$ then
$R^\M(\vec \xi) = R^{\M'}(\vec \xi)$ for any $\vec \xi \in M'^n$. Thus by
induction $\M'$ as constructed satisfies $\neg \psi(\vec a)$ as desired.
\end{proof}

But universal classical validities reduce to propositional tautologies: for each
prime formula $\pi$ with $n$ arguments, and each $n$-tuple $\vec x$ of variables
in the language, form a propositional variable $p_{\pi, \vec x}$. Then a
universal formula $\forall \vec y \psi(\vec y, \vec z)$ with $\psi$
quantifier-free is valid iff $\psi^{\mathrm{prop}}$, the propositional formula
where all instances of prime formulas $\pi(\vec x)$ are replaced by the
propositional variable $p_{\pi, \vec x}$, is a propositional tautology. Since
propositional validity is decidable, this combined with
(\ref{0ftautologycollapse}) yields

\begin{thm}
\label{thm:0F-validity_decidable}
For any first order language, the set of finitely 0F-valid formulas coincides
with the set of countably 0F-valid formulas.
They are both decidable.
Therefore, the set of finitely 0E-satisfiable sentences
coincides with the set of countably 0E-satisfiable
sentences, and they are both decidable.
\end{thm}

\subsection{Monadic Relational Language}
\label{ssec:monadic}
\newcommand{\thetaprop}{\theta^{\mathrm{prop}}}
\newcommand{\ineqset}{\textsc{ineqset}}
\newcommand{\NiceComment}[1]{\Comment{\parbox[t]{.5\linewidth}{#1}}}

Let $\lL$ be a first order language
\begin{itemize}
    \item with no equality,
    \item with no function symbols,
    \item with no relation of arity at least 2, and
    \item with at most a finite number of unary predicates $P_1, P_2, \ldots,
    P_s$.
\end{itemize}
We call $\lL$ a \textit{monadic relational language}.
In this subsection, we
show that for any such language, unrestricted, countable, and finite
$\epsilon$E-satisfiabilites and $\epsilon$F-satisfiabilities are all decidable,
ergo the computability of $\epsilon$E-validities as well.

The essence of the proofs in this
section resides in the fact that each $\epsilon$-model in such a language ``has
only a finite amount of information'': They are partitioned by the monadic
predicates into a finite number of indistinguishable parts, and the measures of
these parts uniquely determine the models up to $\epsilon$-elementary
equivalence.
This observation allows us to reduce these satisfiability problems to linear
programming.

\begin{lma} \label{lma:monadic_finite_model}
Let $\lL$ be monadic relational with unary predicates $P_1, P_2, \ldots, P_s.$
Suppose $\pmd$ is an $\epsilon$-model in $\lL$. Then there is a finite
probability model $\pmodel N E$ such that $\pmd \equiv_\epsilon \pmodel N E.$

Furthermore, we can take the universe to be some subset $N \sbe
\powerset(\indset s)$ and require that
\begin{quote}
for every $a, b \in N$, if $P_l^{\mathcal
N}(a)$ holds when and only when $P_l^{\mathcal N}(b)$ holds, then $a = b$.
\end{quote}
\end{lma}
\begin{proof}
For each $U \sbe \indset s$, define the subset
    $$M_U := \{a \in M: \forall l \in \indset s,\ P_l^\M(a) \iff l \in U\}.$$
$\{M_U\}_{U \sbe \indset s}$ partitions $\pmd$ into at most $2^s$ disjoint
parts. It should be immediate that for any formula $\phi(x)$, any $U \sbe
\indset s$, and $a, b \in M_U$,
\begin{equation}
\pmd \emodels \phi(a) \iff \pmd \emodels \phi(b).
\label{stmt:indist}\tag{$\Delta$}
\end{equation}

Now we define $\pmodel N E$. Let
    $$N := \{U: M_U \not= \emptyset\}$$
and let $\mathcal E$ be defined on points $U$ by
    $$\mathcal E(U) = \D(M_U).$$
Then $\sum_{U \in N} \mathcal E(U) = 1$, so $\mathcal E$ is a probability
measure.

Finally, define the interpretations $P_l^\N$ on $N$ by
    $$P_l^N(U) :\iff l \in U.$$

Since all subsets of $N$ are $\mathcal E$-measurable, $\pmodel N E$ is a
probability model.

For $\epsilon$-elementary equivalence, we show the stronger claim that:
\begin{quote}
For any quantifier-free formula $\phi(\vec x, \vec
y)$ with $j = |\vec x|$ and $k = |\vec y|$, every $\vec U \in N^{k}$, and
every sequence of quantifiers $\varq_1, \ldots, \varq_{j} \in \{\exists,
\forall\}$,
    $$\pmodel N E \emodels \varq_1 x_1 \cdots \varq_{j} x_j \phi(\vec x, \vec U)$$
iff for all (and, by (\ref{stmt:indist}), for any) $\vec a \in M_{\vec U} :=
M_{U_1} \times M_{U_2} \times \cdots \times M_{U_k}$,
    $$\pmd \emodels \varq_1 x_1 \cdots \varq_{j} x_j \phi(\vec x, \vec a).$$
\end{quote}

We proceed by induction on the number $j$ of quantifiers. The
case of $j = 0$ is immediate by our construction.

\let\oldvarqs\varqs
\renewcommand{\varqs}{\oldvarqs{j'}}
Suppose that our claim is proved for $j = j' \ge 0$. The case of
$\psi(\vec U) := \exists y \varqs \phi(\vec x, y, \vec U)$ does not involve
measures and is obvious.
For $\psi(\vec U) := \forall y \varqs \phi(\vec x, y, \vec U)$,
$\pmodel N E \emodels \psi(\vec U)$ iff
    $$W := \{V: \pmodel N E \emodels\varqs \phi(\vec x, V, \vec U)\},\quad
    \mathcal E(W) \ge 1 - \epsilon.$$
By induction hypothesis,
    $$W = \{V: \forall a \in M_V,\ \forall \vec b \in M_{\vec U},\
            \pmd \emodels \varqs \phi(\vec x, a, \vec b)\}.$$

Therefore, for all $\vec b \in M_{\vec U},$
    $$\bigcup_{V \in W} M_V = \{a: \pmd \emodels\varqs \phi(\vec x, a, \vec b)\}$$
and thus
\beq
    &&\Pr_{a \sim \mathcal D}[\pmd \emodels\varqs \phi(\vec x, a, \vec b)]\\
    &=& \Pr_{a \sim \mathcal D}[a \in \bigcup_{V \in W} M_V]\\
    &=& \sum_{V \in W} \D(M_V)\\
    &=& \sum_{V \in W} \mathcal E(V)\\
    &=& \mathcal E(W)\\
    &\ge& 1 - \epsilon.
\eeq
implying
    $$\pmd \emodels \forall y \varqs \phi(\vec x, y, \vec U).$$

The converse direction follows by reversing this line of reasoning and applying
(\ref{stmt:indist}).

The claim starting from ``Furthermore'' follows by our construction.
\end{proof}

We introduce the concept of $\epsilon$E- and $\epsilon$F-trees to help us
analyze satisfiability of sentences.
\newcommand{\child}[2]{{#1}^{\succ #2}}
\newcommand{\rt}{\operatorname{root}}
\begin{defn}
Let $M$ be a set.
A \textbf{tree in $M$ with height $n$} is defined as a tree $T$ with $n$
levels (from 1 to $n$) with the following properties
\begin{enumerate}
    \item all nodes are subsets of $M$.
    \item if node $V$ is at level $k < n$, then $V$ has a child $\child V x$
    for each $x \in V$, and these are all of $V$'s children.
    \item if node $V$ is at level $n$, then $V$ has no children; $V$ is called
    a \textit{leaf node} of $T$.
\end{enumerate}

The unique root of $T$ is denoted $\rt T$.
A \textbf{bran} of a tree in $M$ with height $n$ is defined as a sequence
$\la (a_i, V_i) \ra_{i=1}^n$ of pairs, where for each $i$,
\begin{enumerate}
    \item $V_i \sbe M$ is a node of $T$ at level $i$,
    \item $a_i \in V_i$, and
    \item $V_{i+1} = \child{V_i}{a_i}$ if $i < n$.
\end{enumerate}
We will also write $\la a_i \in V_i \ra_{i=1}^n$ for a bran, which should not
cause any confusion.
\end{defn}

\begin{defn} \label{defn:epsilonE-tree}
Let $Q = \la \varq_1, \varq_2, \ldots, \varq_n\ra$ be a sequence of quantifiers
from $\{\exists, \forall\}$.
Let $\pmd$ be an $\epsilon$-model.
An \textbf{$\epsilon$E-tree in $\pmd$ with levels $Q$} is defined as
a tree in $M$ with height $n$ such that:
\begin{enumerate}
    \item if $\varq_k = \exists$, then all nodes at level $k$ are nonempty
    subsets of $M$; level $k$ is called a $\exists$-level.
    \label{prpty:Etree_exists}
    \item if $\varq_k = \forall$, then all nodes at level $k$ are
    $\D$-measurable subsets of $M$ with $\D$-measure at least $1 -
    \epsilon$;
    level $k$ is called a $\forall$-level.\label{prpty:Etree_forall}
\end{enumerate}

Let $\Phi:=\varqs n \phi(\vec x)$ where $\phi$ is quantifier-free.
An \textbf{$\epsilon$E-tree in $\pmd$ for $\Phi$} is defined as an
$\epsilon$E-tree in $\pmd$ with levels $\la \varq_1, \ldots, \varq_n\ra$ with
the additional property that
\begin{enumerate}[label=(S)]
    \item \label{prpty:Ebranch}
    for every bran $\la a_i \in V_i\ra_{i=1}^n$,
        $$\M \models \phi(\vec a).$$
\end{enumerate}

When $\Phi$ (or $Q$) and $\pmd$ are clear from the context, we will simply use
the term \textbf{$\epsilon$E-tree}.
\end{defn}

\begin{exmp}
If $Q = \la \exists, \forall, \exists\ra$, $M = \indset 4$, and $\D$ is the
uniform distribution, then
$$
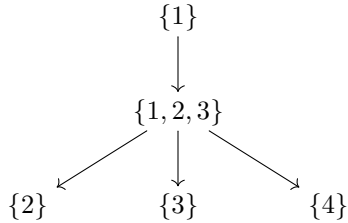
\begin{tikzcd}
        &	\{1\} \dar\\
        &	\{1, 2, 3\} \dlar \dar \drar\\
\{2\}	&	\{3\}	&	\{4\}
\end{tikzcd}$$
is an $\f 1 2$E-tree in $\pmd$ with levels $Q$.

If $\Phi = \exists x \forall y \exists z [x + y = z]$ (where $x + y = z$ should
be treated as a relation over $(x, y, z)$), then the above tree is
also an $\f 1 2$E-tree in $\pmd$ for $\Phi$.

Note that nodes across a level need not be distinct.
For example, if $\Phi = \forall x \exists y [x \not = y]$, then
$$
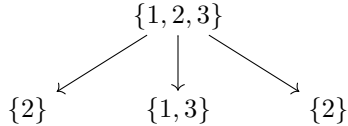
\begin{tikzcd}
        &	\{1, 2, 3\} \dlar \dar \drar\\
\{2\}	&	\{1, 3\}	&	\{2\}
\end{tikzcd}$$
would be a valid $\f 1 2$E-tree in $\pmd$ for $\Phi$.
\end{exmp}

Similarly, we define

\begin{defn} \label{defn:epsilonF-tree}
Let $Q = \la \varq_1, \varq_2, \ldots, \varq_n\ra$ be a sequence of quantifiers
from $\{\exists, \forall\}$.
Let $\pmd$ be an $\epsilon$-model.
An \textbf{$\epsilon$F-tree in $\pmd$ with levels $Q$} is defined as
a tree in $M$ with height $n$ such that:
\begin{enumerate}
    \item if $\varq_k = \exists$, then all nodes at level $k$ are
    $\D$-measurable subsets of $M$ with $\D$-measure greater than $\epsilon$;
    level $k$ is called a $\exists$-level.\label{prpty:Ftree_exists}
    \item if $\varq_k = \forall$, then all nodes at level $k$ are $M$;
    level $k$ is called a $\forall$-level.
    \label{prpty:Ftree_forall}
\end{enumerate}

Let $\Phi:=\varqs n \phi(\vec x)$ where $\phi$ is quantifier-free.
An \textbf{$\epsilon$F-tree in $\pmd$ for $\Phi$} is defined as an
$\epsilon$F-tree in $\pmd$ with levels $\la \varq_1, \ldots, \varq_n\ra$ with
the additional property that
\begin{enumerate}[label=(S)]
    \item \label{prpty:Fbranch}
    for every bran $\la a_i \in V_i \ra_{i=1}^n$,
        $$\M \models \phi(\vec a).$$
\end{enumerate}

When $\Phi$ and $\pmd$ are clear from the context, we will simply use the term
\textbf{$\epsilon$F-tree}.
\end{defn}

\begin{exmp}
With the same setting as above of $Q = \la \exists, \forall,
\exists\ra$, $M = \indset 4$, and $\D$ the uniform distribution,
$$
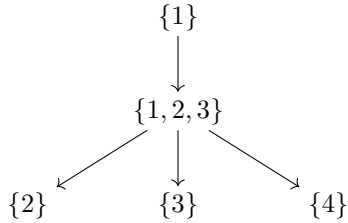
\begin{tikzcd}
        &	\{1\} \dar\\
        &	\{1, 2, 3\} \dlar \dar \drar\\
\{2\}	&	\{3\}	&	\{4\}
\end{tikzcd}$$
is \textit{not} a $\f 1 2$F-tree in $\pmd$ with levels $Q$ because the root
$\{1\}$ has measure $\f 1 4 < \f 1 2$ and the second level node is not all of
$M$.
However, the following \textit{is} a $\f 1 2$F-tree with levels in $\la
\exists, \forall\ra$:
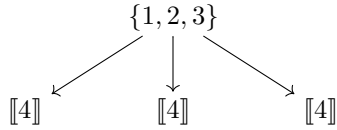
$$\begin{tikzcd}
        &	\{1, 2, 3\} \dlar \dar \drar\\
\indset 4&	\indset 4	&	\indset 4\\
\end{tikzcd}$$

If $\Phi = \forall x \exists y [ x \not= y]$, then
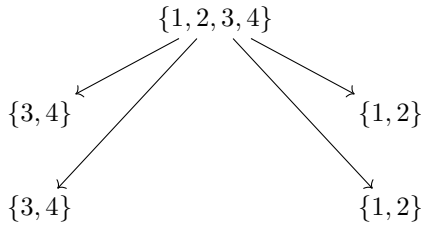
$$\begin{tikzcd}
        &	\{1, 2, 3, 4\} \dlar \ar{ddl} \ar{ddr} \ar{dr}\\
\{3, 4\}&		&\{1, 2\}	\\
\{3, 4\}& 		&\{1, 2\}
\end{tikzcd}$$
is a $\f 1 4$F-tree for $\Phi$, but not a $\f 1 2$F-tree.
As an exercise, the reader should verify that there is no $\f 3 4$F-tree for
$\Phi$ in our choice of $\pmd$.
Note once again that, as this example illustates, nodes across the same
level in an $\epsilon$F-tree need not be distinct.
\end{exmp}

It should be apparent from the definitions and examples that
\begin{prop} \label{prop:sat_reduce2tree}
For any first order language $\lL$, let $\Phi$ be an $\lL$-sentence and $\pmd$
be an $\epsilon$-model of $\lL$.
Then $\pmd \emodels \Phi$ iff there exists an $\epsilon$E-tree in $\pmd$ for
$\Phi$.
Similarly, $\pmd \fmodels \Phi$ iff there exists an $\epsilon$F-tree in $\pmd$ for
$\Phi$.

Therefore, $\Phi$ is (unrestricted/finitely/countably) $\epsilon$E-satisfiable
iff there exists an $\epsilon$E-tree in some
(unrestricted/finite/countable) $\pmd$ for $\Phi$.
This statement holds also when F substitutes E.
\end{prop}

Now we are ready to tackle the decidability of monadic relational languages.

\begin{thm} \label{thm:E-monadic_decidable}
Let $\lL$ be a monadic relational first order language and
$\epsilon \in [0, 1]$ be a rational number.
Then $\epsilon$E-satisfiability in $\lL$ is decidable for the unrestricted, the
countable, and the finite cases.
\end{thm}
\begin{proof}
\newcommand{\esatdecide}{\textsc{esat-decide}}

By lemma (\ref{lma:monadic_finite_model}), all three
$\epsilon$E-satisfiability follow from the finite case, so we will prove the
latter.
In particular, it suffices to consider only
satisfaction by probability models $\pmd$ with universe $M \sbe
\powerset(\indset s)$, with $\D$ everywhere defined, and with the property that
    $$\forall a, b \in M,\ [\forall l \in \indset s, P^\M_l(a) \iff P^\M_l(b)]
        \implies a = b.$$
Call such models \textit{simple models}.

Let $\Phi$ be a sentence in $\lL$, without loss of generality in prenex normal
form
    $$\Phi := \varqs n \phi(\vec x)$$
where $\phi(\vec x)$ is quantifier-free and each $\varq_i$ is a quantifier.

By proposition (\ref{prop:sat_reduce2tree}), $\Phi$ is $\epsilon$E-satisfiable
by simple models iff there is an $\epsilon$E-tree in some simple $\pmd$ for
$\Phi$.
But all such trees are trees in some subset of $\powerset(\indset{s})$ with
height $n$, which are finite in number.
Thus, we only need to devise an algorithm that,
for any $\dot M \sbe \powerset(\indset{s})$ and any tree $T$ in $\dot M$ with
height $n$, tests whether there exists simple model $\pmd$ with universe
$\dot M$ such that $T$ is an $\epsilon$E-tree in $\pmd$ for $\Phi$ ---
and it suffices to verify definition (\ref{defn:epsilonE-tree}).
But $\dot M$ is the universe of some such simple model $\pmd$ iff $\dot M$ is so
with the interpretations
    $$\forall U \in \dot M,\ P_l^\M(U) \iff l \in U.$$
Thus, assuming this structure of $\M$,
we can check conditions (\ref{prpty:Etree_exists}) and \ref{prpty:Ebranch}
of definition (\ref{defn:epsilonE-tree}) immediately, in finite time.

Finally, for any surviving $T$ and $\M$ with universe $\dot M$, we find whether
there exists a probability measure $\D$ everywhere defined such that $T$ satisfies
condition (\ref{prpty:Etree_forall}).
$\pmd$ fulfills this condition iff for every node $V \sbe \dot M$ at a
$\forall$-level in $T$,
\begin{equation}
\sum_{a \in V} \D(a) \ge 1 - \epsilon. 		\tag{$\diamond_V$} \label{ineq:LP}
\end{equation}
This is equivalent, then, to the feasibility of the linear program LP with
variables $\mu_a$ for every $a \in \dot M$ and inequalities
\begin{enumerate}
    \item (\ref{ineq:LP}) with $\mu_a$ replacing $\D(a)$,
    \item $\mu_a \ge 0$, for each $a \in \dot M$, and
    \item $\sum_{a \in \dot M} \mu_a = 1$.
\end{enumerate}

Evidently, because $\epsilon$ is rational, all coefficients in LP are rational.
By proposition (\ref{prop:LPrat_decidable}), LP is solvable in finite time.

Thus condition (\ref{prpty:Etree_forall}) can be verified effectively.

\end{proof}

The corresponding result for $\epsilon$F-logic will proceed similarly, except
that the linear program will involve strict inequalities. The following
result of Carver is needed to prove the next lemma.

\begin{prop}[Carver \cite{schrijver}]
Let $A$ be a matrix and let $b$ be a column vector. There exists
a vector $x$ with $Ax < b$ iff $y = 0$ is the only solution for
    $$y \ge 0,\ yA = 0,\ yb \le 0.$$
\end{prop}

\begin{lma} \label{lma:strict_LP_decide}
Let $A$ and $B$ be matrices and let $b$ and $c$ be column vectors.
If all entries in $A, B, b, c$ are rational, then there is an
algorithm that decides the feasibility of the system
    $$Ax < b,\ Bx = c.$$
\end{lma}
\begin{proof}
Solve $Bx = c$ for $x$ (for example by Gaussian elimination). If there is no
solution $x$, then the system is not feasible. If there is exactly one solution
$x$, we check whether $x$ satisfies $Ax < b$ and return the result. Finally, if
the solution set forms an affine plane of dimension $d$, then there exist $d$
indices $i_1, \ldots, i_d$ such that each coordinate $x_k$ is a linear
combination of $x_{i_1}, \ldots, x_{i_d}$ and 1.
Substituting these equations into $Ax < b$ yields a new (strict) linear program
$A'x' < b'$ with rational coefficients where $x' = (x_{i_1}, x_{i_2}, \ldots,
x_{i_d})$.
Evidently the feasibility of the original system is equivalent to the
feasibility of $A'x' < b'$, which can be solved via Carver's theorem and
proposition (\ref{prop:LPrat_decidable}).
\end{proof}

Now we are ready to characterize the monadic relational fragment of
$\epsilon$F-logic.
The method of proof follows roughly the same path as for theorem
(\ref{thm:E-monadic_decidable}).

\begin{thm} \label{thm:F-monadic_decidable}
Let $\lL$ be a monadic relational first order language
and $\epsilon \in [0, 1)$ be a rational number.
Then $\epsilon$F-satisfiability in $\lL$ is decidable for the unrestricted, the
countable, and the finite cases.
\end{thm}

\newcommand{\choiceset}{\{i_j \in I_j\}_{j=1}^m}
\begin{proof}
Let $\Phi$ be a sentence in prenex normal form.

Again, by lemma (\ref{lma:monadic_finite_model}), it suffices to consider only
the finite case and only
satisfaction by probability models $\pmd$ with universe $M \sbe
\powerset(\indset s)$, with $\D$ everywhere defined, and with the property that
    $$\forall a, b \in M,\ [\forall l \in \indset s, P^\M_l(a) \iff P^\M_l(b)]
        \implies a = b.$$
Call such models \textit{simple models}.

Just as in the proof of theorem (\ref{thm:E-monadic_decidable}), it's enough to
check in finite time, for each $\dot M \sbe \powerset(\indset s)$ and each tree
$T$ in $\dot M$, whether there is a simple $\pmd$ with universe $\dot M$ such
that $T$ is an $\epsilon$F-tree in $\pmd$ for $\Phi$.
If $\pmd$ is some such model, then $T'$ is an $\epsilon$F-tree in
$\pmodel{M'}{D'}$ for $\Phi$, where
\begin{itemize}
    \item $M'$ is the set of nonnull elements of $\dot M$,
    \item $\D'$ is the restriction of $\D$ to $M'$, and
    \item $T'$ is derived from $T$ by restricting every node $V \sbe \dot M$ of
    $T$ to $M'$.
\end{itemize}
Therefore we may consider only $\D$ that is everywhere positive.

Again, we can assume $\M$ to have interpretations
    $$P_l^\M(U) \iff l \in U$$
for each $U \in \dot M$ and $l \in \indset s$.
So conditions (\ref{prpty:Ftree_forall}) and \ref{prpty:Fbranch} of definition
(\ref{defn:epsilonF-tree}) can be easily verified.

Finally, for any surviving $T$ and $\M$ with universe $\dot M$, we find whether
there exists a probability measure $\D$ everywhere defined such that $T$ satisfies
condition (\ref{prpty:Ftree_exists}).
Thus $T$ is an $\epsilon$F-tree in $\pmd$ iff for every node $V \sbe \dot M$ at
an $\exists$-level in $T$,
\begin{equation}
\sum_{a \in V} \D(a) > \epsilon. 		\tag{$\circ_V$} \label{ineq:LP_strict}
\end{equation}
The existence of such $\D$ is equivalent to the feasibility of the strict linear
program LP with variables $\mu_a$ for every $a \in \dot M$ and inequalities
\begin{enumerate}
    \item (\ref{ineq:LP_strict}) with $\mu_a$ replacing $\D(a)$,
    \item $\mu_a > 0$, for each $a \in \dot M$, and
    \item $\sum_{a \in \dot M} \mu_a = 1$.
\end{enumerate}

By lemma (\ref{lma:strict_LP_decide}), this is decidable.

\end{proof}

By duality, we can phrase theorems (\ref{thm:E-monadic_decidable}) and
(\ref{thm:F-monadic_decidable}) thus
\begin{cor}
Let $\lL$ be any monadic relational first order language
and $\epsilon \in (0, 1)$ be a rational number.
Then for X = E or F,
(unrestricted/countable/finite) $\epsilon$X-satisifiability and validity are
both decidable.
\end{cor}

\subsection{q-Sentence, q-Trees, and q-Satisfiability}
\label{ssec:q-sentence}

Here we generalize the various concepts of $\epsilon$E and $\epsilon$F like
trees and satisfiability.
These generalizations will allow us to express computability reduction results
in the next section.

\begin{defn}
For any first order language $\lL$,
a \textbf{q-sentence in $\lL$} is defined as a string of the form
    $$\varqs n \phi(\vec x)$$
where $\phi$ is a quantifier-free $\lL$-formula in $n$ variables, and
for each $i$,
    $$\varq_i \in \Qset := \{\exists, \forall\} \cup
        \{ \Qr\ge\}_{\epsilon \in \Q \cap [0, 1]} \cup
        \{ \Qr>\}_{\epsilon \in \Q \cap [0, 1]}.$$

Quantifiers of the form $\Qr\ge$ are called \textbf{weak q-quantifiers}.
Quantifiers of the form $\Qr>$ are called \textbf{strong q-quantifiers}.

If q-sentence $\Phi$ has no strong q-quantifiers and no $\forall$, then
$\Phi$ is called a \textbf{qE-sentence}.
In the same way, if $\Phi$ has no weak q-quantifiers and no $\exists$, then
$\Phi$ is called a \textbf{qF-sentence}.
\end{defn}

\begin{defn}
Let $\Phi:= \varqs n \phi(\vec x)$ be a first order logic sentence, with
$\varq_i \in \{\exists, \forall\}$ and $\phi$ quantifier-free.
The \textbf{$\epsilon$E-coercion of $\Phi$} is defined as the q-sentence
    $$\epsilon\Ecoerce(\Phi): = \varq'_1 x_1 \cdots \varq'_n x_n \phi(\vec x)$$
where
    $$ \varq'_i = \begin{cases}
            \exists		&	\text{if $\varq_i = \exists$}\\
            \Qr[1 - \epsilon]\ge	&	\text{if $\varq_i = \forall$.}
    \end{cases}$$

Likewise,
the \textbf{$\epsilon$F-coercion of $\Phi$} is defined as the q-sentence
    $$\epsilon\Fcoerce(\Phi): = \varq'_1 x_1 \cdots \varq'_n x_n \phi(\vec x)$$
where
    $$ \varq'_i = \begin{cases}
            \forall		&	\text{if $\varq_i = \forall$}\\
            \Qr>	&	\text{if $\varq_i = \exists$.}
    \end{cases}$$

For an arbitrary first order sentence $\Phi$, its $\epsilon$E-coercion is the
coercion of the equivalent prenex normal form.
Similarly for $\epsilon$F-coercion.

\end{defn}

Clearly, both coercion functions are computable.
In addition, every q-sentence in the image of $\epsilon\Ecoerce$ is a
qE-sentence, and every q-sentence in the image of $\epsilon\Fcoerce$ is a
qF-sentence.

\begin{defn} \label{defn:qtree}
Let $Q = \la \varq_1, \varq_2, \ldots, \varq_n\ra$ be a sequence of quantifiers
from $\Qset$.
Let $\M$ be a first order model and $\D$ be a probability measure on its
universe $M$.
An \textbf{q-tree in $\pmd$ with levels $Q$} is defined as
a tree in $M$ with height $n$ such that:
\begin{enumerate}
    \item if $\varq_k = \exists$, then all nodes at level $k$ are
    nonempty subsets of $M$;
    level $k$ is called a $\exists$-level.\label{prpty:qtree_exists}
    \item if $\varq_k = \forall$, then all nodes at level $k$ are $M$;
    level $k$ is called a $\forall$-level.
    \label{prpty:qtree_forall}
    \item if $\varq_k = \Qr\ge$, then all nodes at level $k$ are
    subsets of $M$ with $\D$-measure at least $\epsilon$;
    level $k$ is called a $\Qr\ge$-level.
    \label{prpty:qtree_weak}
    \item if $\varq_k = \Qr>$, then all nodes at level $k$ are
    subsets of $M$ with $\D$-measure greater than $\epsilon$;
    level $k$ is called a $\Qr>$-level.
    \label{prpty:qtree_strong}
\end{enumerate}

Let $\Phi:=\varqs n \phi(\vec x)$ be a q-sentence.
A \textbf{q-tree in $\pmd$ for $\Phi$} is defined as a
q-tree in $\pmd$ with levels $\la \varq_1, \ldots,
\varq_n\ra$ with the additional property that
\begin{enumerate}[label=(S)]
    \item \label{prpty:qbranch}
    for every bran $\la a_i \in V_i \ra_{i=1}^n$,
        $$\M \models \phi(\vec a).$$
\end{enumerate}

When $\Phi$ and $\pmd$ are clear from the context, we will simply use the term
\textbf{q-tree}.

\end{defn}

\begin{exmp}
Let $M = \indset 4$ and $\D$ be the uniform distribution.
The following q-tree in $\pmd$
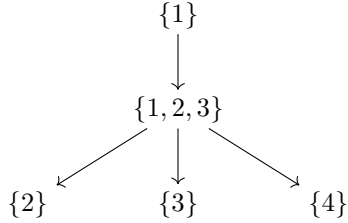
$$\begin{tikzcd}
        &	\{1\} \dar\\
        &	\{1, 2, 3\} \dlar \dar \drar\\
\{2\}	&	\{3\}	&	\{4\}
\end{tikzcd}$$
is a q-tree with levels $Q$ for $Q = \la \exists, \Qr[3/4]\ge, \Qr[0]>\ra$,
$\la \Qr[1/4]\ge, \exists, \exists\ra$, and $\la \Qr[0]>, \Qr[1/2]>,
\Qr[0]\ge\ra$, but not for $Q = \la \exists, \forall, \exists\ra$ or
$\la \Qr[1/2]>, \Qr[3/4]>, \Qr[1/4]>\ra$.
Thus $T$ can be an $\epsilon$E-tree with levels $Q$ but not necessarily be
a q-tree with levels $Q$, as $\forall$ is interpreted differently.

It is also a q-tree for $\Phi$ if
$\Phi = \exists x \Qr[1/2]> y \exists z [x + y = z]$ but not if
$\Phi = \exists x \forall y \exists z [x + y = z]$.

The following q-tree in $\pmd$
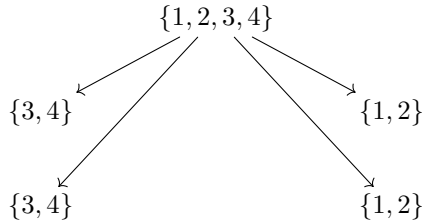
$$\begin{tikzcd}
        &	\{1, 2, 3, 4\} \dlar \ar{ddl} \ar{ddr} \ar{dr}\\
\{3, 4\}&		&\{1, 2\}	\\
\{3, 4\}& 		&\{1, 2\}
\end{tikzcd}$$
is a q-tree with levels $Q$ for $Q = \la \forall, \exists\ra$ and
$\la \Qr[1]\ge, \Qr[1/4]>\ra$ but not for
$Q = \la \exists, \forall\ra$ or $\la \forall, \Qr[1/2]>\ra$.

It is also a q-tree for $\Phi$ if
$\Phi = \forall x \exists y [x \not= y]$ or
$\Qr[1]\ge x \Qr[0]> y [x \not = y]$ but not if
$\Phi = \forall x \forall y [x \not= y]$ or
$\forall x \Qr[3/4]\ge y [x \not= y]$.
\end{exmp}
Clearly, q-trees are generalizations of both $\epsilon$E-trees and
$\epsilon$F-trees:
An $\epsilon$E-tree in $\pmd$ for $\Phi$ is exactly a q-tree in $\pmd$ for
$\epsilon\Ecoerce(\Phi)$.
Likewise for $\epsilon$F-trees.

\begin{defn}
A pair $\pmd$ of first order model $\M$ and probability measure $\D$ on $M$ is
said to \textbf{q-satisfy} a q-sentence $\Phi := \varqs n \phi(\vec x)$,
written
    $$\pmd \qmodels \Phi,$$
if there exists a q-tree in $\pmd$ for $\Phi$.

A q-sentence $\Phi$ is said to be \textbf{q-satisfiable} if some $\pmd$
q-satisfies $\Phi$.
\end{defn}

Again, q-satisfiability is just a generalization of $\epsilon$E- and
$\epsilon$F-satisfibiliy:
for a first order sentence $\Psi$,
    $$\pmd \emodels \Psi \iff \pmd \qmodels \epsilon\Ecoerce(\Psi).$$
The obvious analogue holds for $\epsilon$F-satisfiability as well.

\begin{defn}
Let $Q = \la \varq_1, \ldots, \varq_n\ra$ and $Q' = \la \varq'_1, \ldots,
\varq'_m\ra$.
Let $\pmd$ be a pair of first order model and probability measure.
Suppose $T$ and $T'$ are q-trees in $\pmd$ respectively with levels $Q$ and
$Q'$.
The \textbf{wedge product $T \wedge T'$} is defined as the q-tree
with levels $\la \varq_1, \ldots, \varq_n, \varq'_1, \ldots, \varq'_m\ra$
(thus of height $n + m$) constructed as follows:
\begin{quote}
For every leaf node $V$ in $T$ and every element $a \in V$, set
a copy of $T'$ as the subtree under $a$.
\end{quote}

The expression $T_1 \wedge T_2 \wedge \cdots \wedge T_k$ is parsed as
    $$(((T_1 \wedge T_2) \wedge \cdots )\wedge T_k).$$
\end{defn}

\begin{exmp}
Let $T$ be
$$
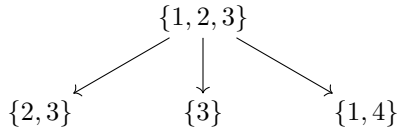
\begin{tikzcd}
        &	\{1, 2, 3\} \dlar \dar \drar\\
\{2, 3\}	&	\{3\}	&	\{1, 4\}
\end{tikzcd}$$
and $T'$ be the singleton tree $\{1\}$. Then $T \wedge T'$ is the q-tree
$$\begin{tikzcd}
        &	&	\{1, 2, 3\} \dlar \dar \drar\\
    &\{2, 3\}\dlar \dar	&	\{3\}\dar	&	\{1, 4\}\dar \drar\\
\{1\}	&	\{1\}	&	\{1\}	&	\{1\}	&	\{1\}
\end{tikzcd}$$

\end{exmp}
\begin{prop}
Let $\Phi := \varqs n \phi(\vec x)$ and $\Phi' := \varqs['y]m \phi'(\vec y)$ be
q-sentences.
Let $\pmd$ be a pair of first order model and probability measure.
Suppose $T$ and $T'$ are q-trees in $\pmd$ respectively for $\Phi$ and $\Phi'$.
Then $T \wedge T'$ is a q-tree in $\pmd$ for the q-sentence
    $$\varqs n \varqs['y]m [\phi(\vec x) \land \phi'(\vec y)].$$
\end{prop}
\begin{proof}
We verify definition (\ref{defn:qtree}).
Conditions (\ref{prpty:qtree_exists}), (\ref{prpty:qtree_forall}),
(\ref{prpty:qtree_weak}), and (\ref{prpty:qtree_strong}) follow easily from the
respective conditions on $T$ and $T'$.

Each bran of $T \wedge T'$ is a concatenation
    $$\la a_1 \in V_1,\ \ldots,\ a_n \in V_n,\ b_1 \in W_1,\ \ldots,\ b_m \in
    W_m\ra $$
of a bran $\la a_i \in V_i\ra_{i=1}^n$ of $T$ and a bran
$\la b_i \in W_i \ra_{i=1}^m$ of $T'$.
Thus
    $$\pmd \qmodels \phi(\vec a) \wedge \phi'(\vec b)$$
by $T$ and $T'$'s property \ref{prpty:qbranch}.
So condition \ref{prpty:qbranch} holds for $T \wedge T'$ as well.
\end{proof}

This proposition immediately yields
\begin{prop} \label{prop:qsat_and}
Let $\Phi$ and $\Phi'$ be as above.
A pair $\pmd$ simultaneously q-satisfies $\Phi$ and $\Phi'$
iff
    $$\pmd \qmodels \varqs n \varqs['y]m [\phi(\vec x) \land \phi'(\vec y)].$$
\end{prop}

\subsection{Finite $\epsilon$E-satisfiability}
\label{ssec:Esat}

In this subsection we will show that $\epsilon$E-satisfiability is
$\Sigma^0_1$-complete for rational $\epsilon$ strictly between 0 and 1.
The main reason that we would like to work with rational $\epsilon$ is the
following set of tools provided by Kuyper and Terwijn:
\newcommand{\esatred}{\mathcal{F}}
\begin{lma}[Kuyper-Terwijn inter-reduction
\cite{model_theory_of_e_logic}\cite{sat_is_sigma_1_1}]
Let
\begin{itemize}
\item $\lL$ be a countable first-order language not containing function symbols
or equality,
\item $\lL'$ be the language obtained by adding an infinite number of unary
predicates to $\lL$, and
\item $\epsilon_0, \epsilon_1$ be rational such that
\begin{enumerate}
  \item $0 \le \epsilon_0 \le \epsilon_1 < 1$, or
  \item $0 < \epsilon_1 \le \epsilon_0 \le 1$.
\end{enumerate}
\end{itemize}

Then there is a computable $f$ mapping $\lL$-sentences to $\lL'$-sentences such
that
$\phi$ is $\epsilon_0$E-satisfiable iff $f(\phi)$ is $\epsilon_1$E-satisfiable.

\label{sat_reduction'}

More generally
\footnote{see \cite[remark 2.14]{model_theory_of_e_logic}.},
this reduction works ``per quantifier'':
For any $\epsilon \in (0, 1) \cap \Q$,
there exists a computable function $\esatred_\epsilon$ mapping
qE-sentences in $\lL$ to
$\lL'$-sentences such that the following are equivalent:
\begin{enumerate}
    \item there exists a pair $\pmd$ such that
        $$\pmd \qmodels \Phi.$$
    \item $\esatred_\epsilon(\Phi)$ is $\epsilon$E-satisfiable.
\end{enumerate}

\end{lma}

Even though the theorem only applies to full satisfiability, the proof works
exactly the same for finite (and countable) satisfiability, because the only
model construction in the proof is the duplication of a given satisfying model a
finite number of times, which preserves finiteness (and countability). Thus

\begin{lma}
Let $\lL$ and $\lL'$ be defined as above.
For any $\epsilon \in (0, 1) \cap \Q$,
there exists a computable function $\esatred_\epsilon$ mapping
qE-sentences in $\lL$ to
$\lL'$-sentences such that the following are equivalent:
\begin{enumerate}
    \item there exists a pair $\pmd$ with $M$ finite such that
        $$\pmd \qmodels \Phi.$$
    \item $\esatred_\epsilon(\Phi)$ is finitely $\epsilon$E-satisfiable.
\end{enumerate}
\label{sat_reduction''}
\end{lma}

In particular, we are interested in the following case
\begin{lma}
Let $\lL$ and $\lL'$ be defined as above and fix rational $\epsilon \in (0, 1)$.
There is a computable function $f_\epsilon$ such that,
for any finite set of rationals $J \sbe \Q \cap [0, 1]$ and $\lL$-sentences
$\{\Psi_\alpha\}_{\alpha \in J}$,
the following are equivalent:
\begin{enumerate}
    \item there exists a pair $\pmd$ such that, for each $\alpha \in
    J$, $\pmd$ is a finite $\alpha$-model and
        $$\pmd \emodelsp \alpha \Psi_\alpha.$$
    \item $f_\epsilon(\{\Psi_\alpha\}_{\alpha \in J})$ is finitely
        $\epsilon$E-satisfiable.
\end{enumerate}
\label{sat_reduction}
\end{lma}
\begin{proof}
We have
    $$\pmd \emodelsp \alpha \Psi_\alpha \iff
        \pmd \qmodels \epsilon\Ecoerce(\Psi_\alpha).$$
By repeated applications of proposition (\ref{prop:qsat_and}), $\pmd$
simultaneously q-satisfies all $\epsilon\Ecoerce(\Psi_\alpha)$ iff
$\pmd \qmodels \Gamma$ for some q-sentence $\Gamma$.
Now apply lemma (\ref{sat_reduction''}).
\end{proof}

Therefore, we could compute the simultaneous E-satisfiability over different
error parameters $\alpha \in [0, 1]$ by using only one fixed $\epsilon \in (0,
1)$. This property of $\epsilon \in (0, 1)$ turns out to be powerful enough to
allow the encoding of the halting set by sentences in $\epsilon$E-logic. It also
distinguishes the case of $\epsilon \in (0, 1)$ from the case of $\epsilon = 0$,
for which lemma (\ref{sat_reduction}) is not applicable: whereas the
E-satisfiability of the latter is decidable by theorem
(\ref{thm:0F-validity_decidable}), that of the former, as will be shown next, is
$\Sigma^0_1$-complete.

We now commence the first half of the completeness proof.

\begin{thm}
Let $\lL$ be any countable first-order language with an infinite number of unary
predicates and at least three binary predicates.
Finite $\epsilon$E-satisfiability for $\lL$-sentences is
$\Sigma^0_1$-hard for rational $1 > \epsilon > 0$.
\label{thm:Esat_Sigma1-hard}
\end{thm}

\begin{proof}
\let\oldforall\forall
\newcommand{\eA}{{\epsilon_0}}
\newcommand{\eB}{{\epsilon_1}}
\newcommand{\T}{\mathrm{T}}
\newcommand{\emdlhalf}{\emodelsp{\f 1 2}}
\newcommand{\Lb}{\mathbf{L}}
\newcommand{\Rb}{\mathbf{R}}
\newcommand{\minc}{\underline{\min}}
\newcommand{\maxc}{\underline{\max}}
\newcommand{\allxinN}[1]{\name{\underset{^{>0}}\oldforall x \in N (#1)}}
\newcommand{\halfpr}[2][y]{\name{\Pr_{#1}[#2] = 1/2}}

The main idea of the proof, as remarked above, is that we want to reduce the
halting problem to finite $\epsilon$E-satisfiability.
Specifically we will show
that there is a reduction from the set of Turing machines that halt on empty
input (which is $\Sigma^0_1$-complete) to the set of finite
$\epsilon$E-satisfiables.
The proof is loosely based on the proof of
Trachtenbrot's theorem in Libkin
\cite[p.~166]{libkin} and the proof of $\Sigma^1_1$-hardness of
$\epsilon$E-satisfaction in Kuyper \cite[Thm.~7.6]{sat_is_sigma_1_1}.
Since by theorem (\ref{sat_reduction}) there is a reduction between finite
$\eA$E- and $\eB$E-satisfiability in $\lL$ for any rational $\eA, \eB \in (0,
1)$, it suffices to establish the case of $\epsilon = \f 1 2$.

Suppose that $M = (Q, \nabla, \delta, q_0, Q_a, Q_r)$ is a single-tape
Turing machine, where
\begin{itemize}
    \item $Q$ is the set of states,
    \item $\nabla$ is the tape alphabet,
    \item $q_0$ is the initial state,
    \item $Q_a$ and $Q_r$ are respectively the sets of accepting and rejecting
    states, and
    \item $\delta: Q \times \nabla \to Q \times \nabla \times \{\Lb, \Rb\}$ is
    the transition function.
\end{itemize}

Since we are only interested in Turing machines with empty input,
$\nabla$ can be assumed WLOG to be $\{0, 1\}$ with 0 representing the blank
symbol.

In what follows, we break into three sections the proof for encoding the halting
of $M$ as a finite $\f 1 2$E-satisfiability problem. Section 1 describes the
first order language used. Section 2 constructs the sentence $\Psi$ which is
finitely $\f 1 2$E-satisfiability iff $M$ halts. Finally section 3 proves that
$\Psi$ indeed has such a property.

\begin{proofpart}[The vocabulary]
We define vocabulary
$$\sigma := \{\minc, \maxc, N(\cdot), \eq, <, R(\cdot, \cdot), T(\cdot, \cdot),
H(\cdot, \cdot), (S_q(\cdot))_{q \in Q}\}.$$

(The constants $\minc, \maxc$ in $\sigma$ can be replaced by unary predicates,
so the theorem as stated will still stand.)

The intuition behind this vocabulary, which will be formalized by the
axioms below, is as follows:
\begin{itemize}
\item Elements satisfying $N$ will be ``roughly" a set of positive measure,
linearly ordered (by $<$) elements that will be our measure of time and space, such that $\minc$ and
$\maxc$ are the minimal and the maximal elements of this chain ---
``Roughly", because using $\epsilon$E-logic, we cannot specify that elements of
measure 0 do not satisfy $N$ (in fact, we cannot say anything about elements of
measure 0).
A nontrivial part of this proof is used to maneuver around these ``phantom
elements".

\item $\eq$ is a binary relation mimicking equality. We avoid using true
equality so that we can use the computable reductions of theorem
(\ref{sat_reduction}).
\item $T(p, t)$, where $p, t \in N$, represents that at time $t$, there
is a 1 at position $p$ on the tape.
\item $H(p, t)$, where $p, t \in N$, represents that at time $t$, the head of
the machine is at position $p$.
\item $S_q(t)$, where $t \in N$, represents that at time $t$, the state of the
machine is $q$.
\item $R$ is an auxiliary relation that is used to force certain measures
to be equal. Its purpose will become more clear over the course of the proof.
\end{itemize}
\end{proofpart}

\begin{proofpart}[The encoding sentence]

In this section we define the sentence $\Psi$ that encodes whether $M$ will
halt.
$\Psi$ will be of the form
$$\Psi := f(\{\T_\alpha\}_{\alpha \in J})$$
where

\begin{itemize}
    \item $J = \{0, \f 1 4, \f 1 2, \f 3 4\}$
    \item $f$ is the reduction function $f_{\epsilon}$ from lemma
    (\ref{sat_reduction}) for $J$ and $\epsilon = \f 1 2$, and
    \item $\T_\alpha$ is a sentence for each $\alpha \in J$.
\end{itemize}
Thus $\Psi$ is finitely $\eA$E-satisfiable iff each $\T_\alpha$ is finitely
$\alpha$E-satisfiable.

For a formula $\phi$ in prenex normal form, we recursively define $\phi^N$
(called {\it $\phi$ relativized to $N$}):
\begin{enumerate}
  \item if $\phi$ is quantifier-free, then $\phi^N = \phi$
  \item if $\phi(\vec y) = \forall x \psi(x, \vec y)$, then $\phi(\vec y)^N =
  \forall x(N(x) \limplies \psi(x, \vec y)^N)$
  \item if $\phi(\vec y) = \exists x \psi(x, \vec y)$, then $\phi(\vec y)^N =
  \exists x (N(x) \land \psi(x, \vec y)^N)$.
\end{enumerate}

\renewcommand{\forall}{\operatorname{\underset{^{+}}\oldforall}}

$\boldsymbol{(\T_0)}$.
$\T_0$ will consist of the conjunction of the following sentences (because
$\oldforall x$ here should be interpreted as ``for almost all $x$'', or, as we
only deal with finite models here, as ``for all $x$ with positive measure,''
we will write $\forall$ for the sake of clarity):
\begin{enumerate}
\item All axioms of equality:
    \begin{enumerate}
        \item $\eq$ is an equivalence relation:
            \begin{align*}
            \forall x &(x \eq x)\\
            \forall x \forall y &(x \eq y \limplies y \eq x)\\
            \forall x \forall y \forall z &(x \eq y \land y \eq z \limplies x
            \eq z)
            \end{align*}
        \item the indiscernability of identicals: for each atomic formula $\pi$,
            $$\pi(a_1, \ldots, a_n) \land \bigmeet_{i = 1}^n a_i \eq b_i
            \limplies \pi(b_1, \ldots, b_n)$$
    \end{enumerate}

\item $<$ is a linear order on all elements of $N$ with nonzero
measure:
\begin{align*}
(\forall x \forall y &(x \eq y \lor x < y \lor y < x))^N\\
(\forall x \forall y &(\neg x \eq y \to (x < y \leftrightarrow \neg y <
x)))^N\\
(\forall x \forall y \forall z &(x < y \land y < z \to x < z))^N
\end{align*}
\label{linearorder0}
\item $\minc$ and $\maxc$ are respectively minimal and maximal in $<$:
\begin{align*}
&N(\minc) \land N(\maxc)\\
&(\forall x (x \eq \minc \lor \minc < x))^N\\
&(\forall x (x \eq \maxc \lor x < \maxc))^N
\end{align*}
\label{mincmaxc0}
\item Initially, $M$ is in state $q_0$, the head is in the first position,
and the tape has all zeros: \label{desc:M_init}
\begin{align*}
&S_{q_0}(\minc)\\
&\forall p (p \eq \minc \leftrightarrow  H(p, \minc))\\
&(\forall p \neg T(p, \minc))^N
\end{align*}
\item For any time $t$, $M$ is in a unique state: \label{desc:unique_state}
$$\left(\forall t\lp\bigvee_{q \in Q} S_q(t) \land \bigwedge_{q, q' \in Q}
\neg(S_q(t) \land S_{q'}(t))\rp\right)^N$$
\item A set of sentences encoding the transition function $\delta$.
\label{desc:trans_fun}

First we define a binary relation $\succ$.
The expression $t' \succ t$ is a shorthand for the conjunction of
\begin{itemize}
    \item ``$t'$ is greater than $t$"
        $$t < t'$$
    \item ``for all $s$, $s$ is less than $t'$ iff $s$ is at most $t$"
        $$\forall s (s < t' \lequiv (s < t \lor s \eq t))$$
    \item ``for all $s$, $s$ is greater than $t$ iff $s$ is at least $t'$"
        $$\forall s (t < s \lequiv (t' < s \lor t' \eq s)).$$
\end{itemize}

Thus $t' \succ t$ says that ``$t'$ is a successor or $t$.''

For any formula $\phi(t, \vec x)$, let $\phi(t + 1, \vec x)$ be
defined as the shorthand for the following relativized implication

    $$[\forall t'(t'\succ t \limplies \phi(t', \vec x))]^N.$$

Hence $\phi(t+1, \vec x)$ states that ``$\phi$ holds for the successor of $t$":

Similarly, let $\phi(t - 1, \vec x)$ be defined as the shorthand for

    $$[\forall t'(t\succ t' \limplies \phi(t', \vec x))]^N.$$

The expression $\phi(t-1, \vec x)$ asserts that ``$\phi$ holds for the
predecessor of $t$."

These shorthands are well-defined when applied to multiple variables. For
example, $\phi(p - 1, t + 1)$ is the shorthand for

    \begin{eqnarray*}
    &           &	[\forall p'(p \succ p' \limplies \phi(p', t + 1))]^N \\
    &\mapsto	&	[\forall p'(p \succ p' \limplies
                            [\forall t'(t' \succ t \limplies
                                        \phi(p', t'))]^N)]^N \\
    &\equiv		&	[\forall p'\forall t'(p \succ p' \land t' \succ t
                            \limplies \phi(p', t'))]^N
    \end{eqnarray*}
where in the middle $\mapsto$ means ``expands into."

Now we turn to the task of encoding the transition function.

\newcommand{\cond}{\textsc{cond}}
\newcommand{\trans}{\textsc{trans}}
Suppose the transition function has rule $\delta(q, w) = (q', w', S)$ for some
$q \in Q; w, w' \in \nabla; S \in \{\Lb, \Rb\}$.
For each $q$ and $w$, we define the sentence
    $$\rho_{q, w} := [\forall p \forall t (\cond(p, t) \limplies
                                            \trans(p, t))]^N$$
where $\cond$ and $\trans$ are constructed as follows.

\begin{enumerate}
    \item $\cond(p, t)$ checks that ``the machine $M$ at time
    $t$ has state $q$, has its head pointing at cell $p$, and the character
    under the head is $w$'':
    explicitly, $\cond(p, t)$ is the conjunction of the following:

        \begin{enumerate}
            \item ``The state of $M$ is $q$ at time $t$"

                $$S_q(t)$$

            \item ``The head of $M$ is above cell $p$"

                $$H(p, t)$$

            \item ``The character at cell $p$ is $w$" (exactly one of the
            following sentences belongs to the conjunction, depending on which
            condition is satisfied)

               $$\begin{cases}
               T(p, t) & \text{if $w = 1$}\\
               \neg T(p, t) & \text{if $w = 0$}
               \end{cases}$$

        \end{enumerate}

    \item $\trans(p, t)$ asserts that ``$M$ at time $t + 1$ has state
    $q'$ and has moved $S$ from cell $p$; the cell at $p$ now contains the
    symbol $w'$'': explicitly, $\trans(p, t)$ is the conjunction of the
    following (in every set of alternatives, exactly one of the sentences
    belongs to the conjunction, depending on which condition is satisfied):

        \begin{enumerate}
            \item ``The state of $M$ is $q'$ at time $t + 1$"
                $$S_{q'}(t + 1)$$
            \item ``At time $t + 1$:
            If $S = \Rb$, the head of $M$ is at cell $p + 1$.
            If $S = \Lb$ and $p = \minc$, the head of $M$ is at cell $\minc$.
            Otherwise, the head of $M$ is at cell $p - 1$." \label{headmove}
                $$\begin{cases}
                H(p + 1, t + 1)		&	\text{if $S = \Rb$}\\
                (p \eq \minc \limplies H(p, t + 1)) \land
                (\neg p \eq \minc \limplies H(p + 1, t + 1)
                                    &	\text{if $S = \Lb$}
                \end{cases}$$
            \item ``At time $t + 1$, cell $p$ contains symbol $w'$''
                $$\begin{cases}
                T(p, t + 1)		& \text{if $w' = 1$}\\
                \neg T(p, t + 1)& \text{if $w' = 0$}
                \end{cases}$$
            \item ``All cells other than those involved in (\ref{headmove}) are
            unaffected''. \label{desc:locality}
                $$\begin{cases}
                \forall p'[\neg p' \eq p \land \neg p'\succ p
                            \limplies
                            (T(p', t) \lequiv T(p', t + 1))]
                                        &	\text{if $S = \Rb$}\\
                \forall p'[\neg p' \eq p \land \neg p \succ p'
                            \limplies
                            (T(p', t) \lequiv T(p', t + 1))]
                                        &	\text{if $S = \Lb$}
                \end{cases}$$
        \end{enumerate}

\end{enumerate}

\item We assert that at time $\maxc$, $M$ arrives at an accepting or rejecting
state: \label{desc:halt}
$$\bigvee_{q \in Q_a \cup Q_r} S_q(\maxc)$$

\item For reasons that will become clear later, for all $x, y$ of positive
measure, we need to let $R(x, y)$ hold only when $y$ is not in $N$:

$$\forall x \forall y (R(x, y) \to \neg N(y))$$ \label{R_then_not_N}

\end{enumerate}
This finishes the description of the sentence $\T_0$.
\renewcommand{\forall}{\oldforall}

$\boldsymbol{(\T_{\f 1 4} \textbf{ and } \T_{\f 3 4})}$.
$\T_{\f 1 4}$ and $\T_{\f 3 4}$ are respectively the sentences $\forall x (x =
\minc)$ and $\forall x (x \not= \minc)$. $f_{\f 1 4}(\T_{\f 1 4})$ and $f_{\f 3
4}(\T_{\f 3 4})$ are simultaneously $\f 1 2$E-satisfiable iff the measure of
$\minc$ is at least $\f 3 4$ and the measure of all other elements is
at least $\f 1 4$. Thus $\T_{\f 1 4}$ and $\T_{\f 3 4}$ force the measure of
$\minc$ to be exactly $\f 1 4$.

$\boldsymbol{(\T_{\f 1 2})}$.
$\T_{\f 1 2}$ is the following conjunction of sentences
\renewcommand{\forall}{\operatorname{\underset{^{\ge 1/2}}\oldforall}}
(because the $\oldforall x$ here
should be interpreted as ``for a set of $x$ with measure at least $\f 1 2$'',
we write $\forall$ for the sake of clarity):
\begin{enumerate}
    \item the set of elements in $N$ takes up measure exactly $\f 1 2$:

        $$\forall x N(x) \land \forall x \neg N(x)$$ \label{halfmeasureN}

    \item we want each element of $N$ to have measure equal to the
    total measure of all greater elements.
    This is where the padding relation $R$ is used.

    First, for $\phi(x)$ a formula of a single free variable $x$, we define the
    following shorthand

        $$\allxinN{\phi(x)} := \forall x (N(x) \wedge \phi(x))$$

    The RHS says ``for a set $X$ of measure at least $1/2$, $X \sbe N$ and
    every element of $X$ satisfies $\phi$." If we assume that
    (\ref{halfmeasureN}) is $\f 1 2$E-satisfied, then $\allxinN{\phi(x)}$ is
    equivalent to ``all elements in $N$ of positive measure must satisfy
    $\phi$."

    Secondly, for formula $\psi(y, \vec x)$, we define the following shorthand

        $$\halfpr{\psi(y, \vec x)} := [\forall y \psi(y, \vec x)]
        \land [\forall y' \neg\psi(y', \vec x)]$$

    The conjuncts on the right respectively assert that ``the probability of $y$
    satisfying $\psi(y, \vec x)$ is at least $1/2$" and ``the probability of
    $y'$ not satisfying $\psi(y', \vec x)$ is at least $1/2$." Hence
    $\halfpr{\psi(y, \vec x)}$ says that ``the probability of $y$
    satisfying $\psi(y, \vec x)$ is exactly $1/2$.

    Finally, we define the actual sentences in the conjunction of $\T_{\f 1 2}$:

    (Recall that $\land$ has precedence over $\lor$ for parsing)
    \begin{align*}
    \allxinN{&x \eq \maxc \lor \halfpr{R(x, y) \lor N(y) \land x < y}}\\
    \allxinN{&x \eq \maxc \lor \halfpr{R(x, y) \lor x \eq y}}
    \end{align*}

    If clause $(\ref{R_then_not_N})$ of $\T_0$ is 0E-satisfied (``If $R(x, y)$
    holds then $y$ has measure 0 or is not in $N$"), then the disjuncts $R(x,
    y)$ and $N(y) \land x < y$ are disjoint. Therefore, the above two sentences
    together express ``for all $x$ in $N$ that's not $\maxc$ and not null, the
    measure of the strict final segment of $x$ is the same as the probability of
    $x$ itself; they are both $(\f 1 2 - \Pr_y[R(x, y)])$." Forcing the
    equality of these two measures is the purpose of the predicate $R(\cdot,
    \cdot)$, which has otherwise no uses.
    \label{pr_y[x = y] = pr_y[x < y]}

\end{enumerate}

This concludes the construction of the reducing sentence $\Psi$.
\renewcommand{\forall}{\oldforall}
\end{proofpart}

\begin{proofpart}[The reduction]
\newcommand{\uni}{W}
\newcommand{\md}{\mathcal{\uni}}
\newcommand{\posN}{{N_0^\md}}

Now we show that $M$ halts if and only if $\Psi$ has a finite
E-satisfying model.

\textbf{($M$ halts $\implies$ $\Psi$ satisfiable).}
Suppose $M$ halts in time $m$. We define the
finite satisfying model $(\md, \D)$ thus:
\begin{enumerate}[label=$\circ$]
    \item Let the universe $\uni$ be the set $\{1, 2, \ldots, 2m\}$.
    \item Let $\eq^\md$ be true equality $=$.
    \item Let $N^\md$ be $\{1, 2, \ldots, m\}$
    \item Let $a <^\md b$ for $a, b \in W$ be defined to agree with the
    natural ordering on $\uni$.
    \item Let $\minc^\md = 1$ and $\maxc^\md = m$.
    \item Define the measure $\D(i) = \D(i+m) = 2^{-i-1}$ for $i \in [1, m - 1]$,
    and define $\D(m) = \D(2m) = 2^{-m}$.
    \item Define $R^\md(i, j)$ iff $i \ge j - m \ge 1$
    \item Define $H^\md(i, j)$ iff $j \le m$ and $M$'s head is at position $i$
    at time $j$.
    \item Define $T^\md(i, j)$ iff $j \le m$ and the tape's symbol at position
    $i$ at time $j$ is 1.
    \item Define $S_q^\md(i)$ iff $M$ is in state $q$ at time $i$.
\end{enumerate}

Since all elements of $\md$ have positive measure, all $\forall$ quantifiers in
$\T_0$ are interpreted classically. Therefore one can easily check that
$(\md, \D)$ 0E-satisfies $\T_0$.

As $\minc = 1$ has measure $\f 1 4$, $\T_{\f 1 4}$ and $\T_{\f 3 4}$ are
satisfied.

$N^\md$ obviously has measure $1/2$, so the first clause of $\T_{\f 1 2}$ is
satisfied.

Finally, consider clauses (\ref{pr_y[x = y] = pr_y[x < y]}) in
$\T_{\f 1 2}$.

For a fixed $i < m$, the measure of \{$j:
R^\md(i, j)$\} is $\sum_{k=1}^{i} 2^{-k-1} = 2^{-1} - 2^{-i-1}$, and the measure
of $j \in N^\md$ such that $i < j$ is $\sum_{k=i+1}^{m-1} 2^{-k-1} + 2^{-m} =
2^{-i-1}$. Thus, for this fixed $i$,

    $$\Pr_{j \sim \D}[R(i, j) \lor N(j) \land i < j] = \f 1 2.$$

Letting $i$ vary, we can conclude that ``for all $x$ in $\posN$, either $x$ is
$\maxc^\md$ or the probability of $y$ such that $R(x, y) \lor N(y) \land x < y$
holds is exactly $1/2$." In other words, the following clause in $\T_{\f 1 2}$

    $$\allxinN{x \eq \maxc \lor \halfpr{R(x, y) \lor N(y) \land x < y}}$$
holds in $(\md, \D)$.

Similarly, as $\D(i) = 2^{-i-1}$ for any $i < m$,

    $$\Pr_{j \sim \D}[R(i, j) \lor i \eq j] = \f 1 2,$$
so the following clause in $\T_{\f 1 2}$

    $$\allxinN{x \eq \maxc \lor \halfpr{R(x, y) \lor x \eq y}}$$
is $\f 1 2$E-satisfied by $(\md, \D)$.

Therefore $(\md, \D)$ $\f 1 2$E-satisfies all of $\T_{\f 1 2}$, as desired.

\textbf{($M$ halts $\Longleftarrow$ $\Psi$ satisfiable).}
Let $(\md, \D)$ be a finite $\f 1 2$-model. By
proposition (\ref{finite_model_extension}), we can assume $\D$ is defined on all
subsets of $\md$ and is a probability model. Suppose $(\md, \D) \emdlhalf \Psi$.
We wish to show that the Turing machine $M$ halts. Our strategy will be to show
that $\minc^\md$ and $\maxc^\md$ have positive measures, and every element
between them has positive measure. Then checking the sentences $\T_0$ encoding
$M$ becomes straightforward, as $\forall$ is interpreted classically on this
linear chain.

Let $\posN$ be the subset of elements of $N^\md$ of positive measure.
By the axioms of equality of $\T_0$, the restriction $\eq_0^\md$
of $\eq^\md$ to $\posN$ is an equivalence relation and satisfies the
indiscernability of identicals.
Therefore $\md \equiv \md/\!\!\eq_0^\md$ as first order models, and
for all $\epsilon$,
$(\md, \D)$ is $\epsilon$-elementarily equivalent to the probability model
$(\md/\!\!\eq_0^\md, \D')$, where $\D'$ is defined by assigning each
equivalence class $[a]$ (as a point of $\md/\!\!\eq_0^\md$) the
$\D$-measure of $[a]$ as a set. We are thus justified in assuming that $\eq$ is
true equality on $\posN$ henceforth.

By $\T_{\f 1 4}$ and $\T_{\f 3 4}$ we know that $\D(\minc^\md) = \f 1 4$, so
$\minc^\md \in \posN$. By clause (\ref{halfmeasureN}) of $\T_{\f 1 2}$ we know
that $\D(N^\md) = \f 1 2$.

Any $\forall$ quantifier relativized to $N$ can be interpreted
classically on $\posN$. Thus, by clauses (\ref{linearorder0}) of
$\T_0$, $<^\md$ defines a linear order on $\posN$, and, by clauses
(\ref{mincmaxc0}) of $\T_0$, $\minc^\md$ is the unique minimal element. However,
it is not immediate whether $\maxc^\md$ has positive measure and thus is the
unique maximal element of $\posN$.

But by clause (\ref{R_then_not_N}) of $\T_0$, we have $(\md, \D) \emodelz
\forall x \forall y (R(x, y) \to \neg N(y))$.
Then, for a fixed $x$ of positive measure, the set of $y$ where $R^\md(x,
y)$ holds intersects $N^\md$ with measure 0. Thus, by clauses
(\ref{pr_y[x = y] = pr_y[x < y]}) of $\T_{\f 1 2}$, if $a \in \posN$ is not
$\maxc^\md$, then the equations
\begin{align*}
    \Pr_{y\sim \D}[R(a, y)] + \Pr_{y \sim \D}[N(y) \land a < y] & = \f 1 2
        \qquad \text{and}\\
    \Pr_{y\sim \D}[R(a, y)] + \Pr_{y \sim \D}[a \eq y] & = \f 1 2
\end{align*}
hold.
For such an $a$,
    \begin{equation}\label{D(y) = pr_y[x < y]}
    \Pr_{y \sim D}[N(y) \land a < y] = \Pr_{y \sim \D}[a \eq y] = \D(a).
    \tag{$\star$}
    \end{equation}

\begin{lma}
$\maxc^\md \in \posN$
\end{lma}

\begin{proof}
\newcommand{\two}{\underline{2}^\md}

Let $\two$ denote the immediate successor of $\minc^\md$ in $\posN$; it exists
since $1/4 = \D(\minc^\md) < \D(N^\md) = 1/2$ and $\minc^\md$ is minimal in
the finite set $\posN$.
Suppose $\two \not= \max^\md$, and
    $$\tau := \D(\two) \quad \text{ and } \quad
    \xi := \Pr_{y \sim \D}[N(y) \land \two < y].$$
Then we have
\begin{align*}
    \tau + \xi = \Pr_{y \sim \D}[N(y) \land \minc^\md < y]
                        &= \D(\minc^\md) = \f 1 4,	&&	\text{and}\\
    \tau 				&= \xi,						&&	\text{by
                                                    (\ref{D(y) = pr_y[x < y]}).}
\end{align*}

Hence $\D(\two) = \tau = \xi = \f 1 8$.

\newcommand{\msucc}[1][m]{\underline{#1}^\md}
In general, if the $(n-1)$-fold successor $\msucc[n]$ of $\minc^\md$ has
probability $2^{-n-1}$ for each $n < m$, then the $(m-1)$-fold successor
$\msucc$ of $\minc^\md$ exists, and a) is either $\maxc^\md$, or b) satisfy the following equations
\begin{align*}
    \Pr_{y \in \D}[N(y) \land \msucc[m-1] < y] &= \Pr_{y \in \D}[N(y) \land
    \msucc[m] < y] + \D(\msucc) && \text{and}\\
    \D(\msucc) &= \Pr_{y \in \D}[N(y) \land \msucc < y] &&
                                                    \text{by
                                                    (\ref{D(y) = pr_y[x < y]}),}
\end{align*}
which, along with
\begin{eqnarray*}
  &  & 1/2 - 2^{-m} + \Pr_{y \in \D}[N(y) \land \msucc[m-1] < y]\\
  & = & \D(\minc) + \sum_{n=2}^{m-1} \D(\msucc[n]) + \Pr_{y \in \D}[N(y) \land
  \msucc[m-1] < y]\\
  & = & \Pr_{y \in \D}[N(y)] = 1/2,
\end{eqnarray*}
imply $\Pr_{y \in \D}[N(y) \land \msucc[m-1] < y] = 2^{-m}$ and thus
$\D(\msucc) = 2^{-m-1}$.

Now assume for the sake of contradiction that $\D(\maxc^\md) = 0$. Then $\msucc
= \maxc^\md$ for no finite $m$. This would mean that the $m$-fold successor of
$\minc^\md$ exists for all finite $m$. But $\posN$ is finite, so this cannot be
true. Therefore $\maxc^\md$ must have positive measure, as desired.
\renewcommand{\qedsymbol}{$\blacksquare$}
\end{proof}

We have thus shown that $\posN$ is a linear chain ordered by $<^\md$, with
minimal element $\minc^\md$ and maximal element $\maxc^\md$. This structure
allows us to interpret the sentences in $\T_0$ encoding the Turing machine $M$
classically.

Indeed, we can construct the computation history of $M$ as follows:
\begin{quote}
At time $t$,
\begin{enumerate}[label=---]
    \item the tape has a 1 at position $p$ iff $T(p, t)$ holds,
    \item the head of $M$ is above cell $p$ iff $H(p, t)$ holds, and
    \item the state of $M$ is $q$ iff $S_q(t)$ holds.
\end{enumerate}
\end{quote}

Now, aided by the verbal translation provided in the description of $\T_0$, we
can verify that
\begin{enumerate}[label=$\circ$]
    \item Initially, $M$ is in state $q_0$, the head is in the first position,
    and the tape has all 0s (clauses (\ref{desc:M_init}) of $\T_0$).
    \item At any time, $M$ is in one state and one state only
        (clause (\ref{desc:unique_state}) of $\T_0$).
    \item The tape and $M$'s head position and state are updated correctly
    according to $\delta$ (clauses (\ref{desc:trans_fun}) of $\T_0$).
    In particular, all cells not specified by the update
    rule have the same symbol after the update (clauses (\ref{desc:locality})).
    \item At time $\maxc^\md$, $M$ is in either an accepting or rejecting state
    (clause (\ref{desc:halt}) of $\T_0$).
\end{enumerate}

Hence, the $\f 1 2$E-satisfaction of $\Psi$ implies $M$ halts.
\end{proofpart}
\end{proof}

Finally, to complete our proof that finite $\epsilon$E-satisfiability is
$\Sigma^0_1$-complete, we prove

\begin{thm}
For both X = E and X = F, finite $\epsilon$X-satisfiability is
$\Sigma^0_1$-definable for rational $\epsilon \in (0, 1)$ and any first order
language.
\label{emodels_finite_sat_is_sigma_1}
\end{thm}

Before the proof, we will need the following perturbation results, which allow
us to ``shake up'' the probability measures of each $\epsilon$-model into a
nicer form.

\begin{lma}
Suppose $(M, \dom \D, \D)$ is a measure space such that $M$ is finite and $\D$
is defined on all subsets of $M$. Then for any $\delta > 0$, there exists a measure
$\D'$ defined on all subsets of $M$ such that
\begin{itemize}
  \item all values of $\D'$ are rational,
  \item $\D(S) = \D'(S)$ whenever $\D(S)$ is rational and positive,
  \item $\max_{S \subseteq M}|\D(S) - \D'(S)| < \delta$, and
  \item $\D(S) > 0$ for all $S \sbe M$
\end{itemize}
\label{rat_pert}
\end{lma}

\begin{proof}
\newcommand{\vv}{\textbf{v}}
\newcommand{\pp}{\textbf{p}}
\newcommand{\qq}{\textbf{q}}
\newcommand{\RR}{\textbf{R}}

WLOG let $M = \{1, 2, \ldots, m\}$. Then $\D$ is uniquely determined by its
values on $i \in M$. Let $\pp = \la p_i\ra_{i=1}^m$ represent this vector. Thus
$\D(S) = \vv_S \cdot \pp$ where $\vv_S$ is the vector whose value at position
$i$ is 1 if $i \in S$ and 0 otherwise.

Let $\la S_j \ra_{j=1}^k$ be an enumeration of all $S \subseteq M$ such that
$\D(S)$ is rational and positive. We can then form the matrix $\RR$ with row
vectors $\vv_{S_j}$ and the column vector $\qq = \la \D(S_j)\ra_{j=1}^k$.
Immediately, we have
    $$\RR \pp = \qq.$$

Since all entries of $\RR$ and $\qq$ are rational, by Gaussian elimination, we
can reduce the associated matrix $\RR | \qq$ to row echelon form $\RR' | \qq'$
with all rational entries. It's then clear that we can perturb each value of
$\pp$ by less than $\delta / |M|$ to get $\pp'$ such that 1) each entry of
$\pp'$ is positive and rational, and 2) $\RR' \pp' = \qq'$ and thus $\RR \pp' =
\qq$ and all positive rational values of $\D$ are unaffected.
Extending the point measure $\pp'$ linearly to a measure over all subsets of $M$
gives the desired result.
\end{proof}

\begin{lma}
Let $\epsilon \in (0, 1)$ be rational. If $\pmd$ is a finite
$\epsilon$-model, then for some measure $\D'$ with $\dom \D'=
\powerset(M)$ such that $\D'(x)$ is rational and positive for all $x \in M$,
$\pmodel M {D'}$ is $\epsilon$-elementarily equivalent to $\pmd$.
\label{rational_finite_model}
\end{lma}

\begin{proof}
By proposition (\ref{finite_model_extension}) we may assume that $\D$ is defined
on all subsets of $M$. Using lemma (\ref{rat_pert}) with
    $$\delta = \f 1 2\min\lp|1 - \epsilon|, \min_{\substack{S \subseteq M\\
                            \D(S) \not \in \Q}}
                |\D(S) - (1 - \epsilon)|\rp > 0,$$
there is a measure $\D'$ defined on all subsets of $M$ such that $\D'$ has
all rational and positive values, $\D'$ differs from $\D$ only on sets of
irrational $\D$-measure, and this difference is uniformly bounded by $\delta$. In
particular, $\D(S) \ge 1 - \epsilon \iff \D'(S) \ge 1 - \epsilon$, so (by an
easy induction argument) $\pmd$ is $\epsilon$-elementary equivalent to $\pmodel
M {D'}$. Therefore $\pmodel M {D'} \emodels \phi$.
\end{proof}

For the following proof, we do not actually need $\D'(x)$ to be positive, but
this lemma provides an alternative justification for the assumption in
the proof of (\ref{thm:F-monadic_decidable}).

\begin{proof}[Proof of thm~(\ref{emodels_finite_sat_is_sigma_1})]
\newcommand{\bt}{\textbf{t}}
Let $\la \bt^i \ra_{i\ge 1}$ be an effective enumeration of all finite sequences
of $\omega$ (for example by G\"odel's $\beta$ function). For each finite
sequence $\bt^i$ let $|\bt^i|$ denote the length of the sequence, let $\bt^i_j$
denote the $j$th element of $\bt^i$, for $1 \le j \le |\bt^i|$, and let
$\|\bt^i\|$ denote the sum of its elements $\sum_{j=1}^{|\bt^i|} \bt^i_j$.

Let $\sigma$ be the vocabulary used in $\phi$.
To test whether a sentence $\phi$ is $\epsilon$E-satisfiable by a finite model,
we inspect $\bt^i$ in sequence for $i = 1, 2, \ldots.$ For each $\bt^i$ we
form all classical models with signature $\sigma$ on $|\bt^i|$ elements $\{1, 2, \ldots,
|\bt^i|\}$. There are only a finite number of them since $\sigma$ is finite. We
turn these classical models into $\epsilon$-models by imbuing them with the
measure $\D$ defined by $\D(j) = \bt^i_j/\|\bt^i \|$. We can then mechanically check whether any of
them satisfy $\phi$.

If $\pmd \emodels \phi$ for some finite $\epsilon$-model $\pmd$, then WLOG we
can assume $M = \{1, 2, \ldots, m\}$ for some $m$, $\D$ to be
defined on all subsets of $\M$ (by proposition (\ref{finite_model_extension})),
and $\D$ to have all rational values (by lemma (\ref{rational_finite_model})).
Thus there exists an integer $r$ such that for each $1 \le i \le m$, $\D(i) =
r_i/r$ and $r_i$ is an integer. Our algorithm described above will then
terminate before $k+1$ outer loops, where $\bt^k = \la r_i\ra_{i=1}^m$.

This shows that the $\epsilon$E-satisfiability problem for rational $\epsilon$
is $\Sigma^0_1$-definable. The $\epsilon$F-satisfiability case is resolved
identically.
\end{proof}

Combined with theorem (\ref{thm:Esat_Sigma1-hard}), this result shows
\begin{cor}
Let $\lL$ be a countable first order language with an infinite number of unary
relations and at least three binary relations.
For any rational $\epsilon \in (0, 1)$, finite $\epsilon$E-satisfiability is
$\Sigma^0_1$-complete.
Equivalently, finite $\epsilon$F-validity is $\Pi^0_1$-complete.
\label{thm:Esat_sigma1_complete}
\end{cor}

\subsection{Finite and Countable $\epsilon$E-validity}
\label{ssec:eval}

Like the unrestricted case, the set of finitely (resp. countably) 0E-valid
sentences also coincides with the set of finitely (resp. countably) classically
valid sentences.
With Kuyper's inter-reduction result for $\epsilon$E-validities and the
$\Sigma^0_1$-definability derived from the last section, we can characterize
$\epsilon$E-validity over finite models precisely as $\Pi^0_1$-complete
whenever $\epsilon \in \Q$.

\begin{thm}
For any countable first order language, the set of finitely 0E-valid sentences
is exactly the set of finitely classically valid sentences.
The set of countably 0E-valid sentences is exactly the set of
classically valid sentences.
\end{thm}

\begin{proof}
Obviously every finitely classically valid sentence is a finitely 0E-valid
sentence.

Now suppose $\phi$ is finitely 0E-valid. Then for any $n$, $\pmd
\emodelz \phi$ for all classical models $\M$ of size $n$ and $\D$ the uniform
distribution. But in such 0-models, $\forall$ has the same
interpretation as classically. Hence all finite classical models satisfy $\phi$,
as desired.

The proof works the same for countable validities, except that for countably
infinite models, we instead ascribe the exponential distribution $\D(n) = \f 1
{2^n}$.
\end{proof}

\newcommand{\frag}{\mathcal{S}}
Immediately,
\begin{cor} \label{cor:0Eval=classic}
Let $\lL$ be any first order language and $\frag$ be a set of sentences in
$\lL$.
The following are equivalent:
\begin{itemize}
    \item The set of finitely (resp. unrestricted) classically valid sentences
    in $\frag$ is decidable.
    \item The set of finitely (resp. countably) 0E-valid sentences in
    $\frag$ is decidable.
\end{itemize}
\end{cor}

In any first order language with at least one binary relation,
the set of finitely classically valid sentences is $\Pi^0_1$-complete
\cite[p.~166]{libkin}.
Therefore,
\begin{cor}
In any first order language with at least one binary relation,
the problem of determining whether
a sentence is finitely 0E-valid is $\Pi^0_1$-complete.
\label{fin_eval_Pi1_complete}
\end{cor}

Likewise, as classical validity in any language with at least one binary
relation is $\Sigma^0_1$-complete \cite{Turing}, we have in the countable case
\begin{cor}
\label{cor:countabl_0val_Sigma1_complete}
In any first order language with at least one binary relation,
the problem of determining whether
a sentence is countably 0E-valid is $\Sigma^0_1$-complete.
\end{cor}

There exist computable reductions for $\epsilon$E-validity just like in the case
of $\epsilon$E-satisfiability (see proposition (\ref{sat_reduction'})):

\begin{prop}[Kuyper inter-reduction \cite{val_is_pi_1_1}]
\label{prop:fsat_reduction'}
Let
\begin{itemize}
\item $\lL$ be a countable first-order language not containing function symbols
or equality, and
\item $\lL'$ be the language obtained by adding an infinite number of unary
predicates to $\lL$.
\end{itemize}
Then, for all rational $0 \le \epsilon_0 \le \epsilon_1 < 1$, the set of normally
$\epsilon_0$E-valid $\lL$-sentences many-one reduces via a computable function
to the set of normally $\epsilon_1$E-valid $\lL$-sentences.

More generally
\footnote{see \cite[remark below thm~3.3]{val_is_pi_1_1}.},
this reduction works ``per quantifier'':
For any $\epsilon \in (0, 1) \cap \Q$,
there exists a computable function $\esatred_\epsilon$ mapping
qF-sentences in $\lL$ to
$\lL'$-sentences such that the following are equivalent:
\begin{enumerate}
    \item there exists a pair $\pmd$ such that
        $$\pmd \qmodels \Phi.$$
    \item $\esatred_\epsilon(\Phi)$ is $\epsilon$F-satisfiable.
\end{enumerate}

\end{prop}

Again, the proof for this theorem and the construction of the reduction function
given in \cite{val_is_pi_1_1} carry over almost identically when we restrict our
attention from ``normally $\epsilon$E-valid''
(validity over all probability models)
to ``finitely $\epsilon$E-valid"
(validity over all finite probability models --- which is equivalent to the
validity over all finite models by (\ref{finite_model_extension})):
the method of proof is the duplication of a given model a finite number of
times, and this procedure preserves finiteness of $\epsilon$-models.
This observation remains true in considering countable $\epsilon$-models.

Therefore, we have, by duality, the following two $\epsilon$F-analogues of
(\ref{sat_reduction''}) and (\ref{sat_reduction}).

\begin{lma}
Let $\lL$ and $\lL'$ be defined as above.
For any $\epsilon \in (0, 1) \cap \Q$,
there exists a computable function $\esatred_\epsilon$ mapping
qF-sentences in $\lL$ to
$\lL'$-sentences such that the following are equivalent:
\begin{enumerate}
    \item there exists a pair $\pmd$ with $M$ finite such that
        $$\pmd \qmodels \Phi.$$
    \item $\esatred_\epsilon(\Phi)$ is finitely $\epsilon$F-satisfiable.
\end{enumerate}
\label{fsat_reduction''}
\end{lma}

\begin{lma}
Let $\lL$ and $\lL'$ be defined as above and fix rational $\epsilon \in (0, 1)$.
There is a computable function $f_\epsilon$ such that,
for any finite set of rationals $J \sbe \Q \cap [0, 1]$ and $\lL$-sentences
$\{\Psi_\alpha\}_{\alpha \in J}$,
the following are equivalent:
\begin{enumerate}
    \item there exists a pair $\pmd$ such that, for each $\alpha \in
    J$, $\pmd$ is a finite $\alpha$-model and
        $$\pmd \fmodelsp \alpha \Psi_\alpha.$$
    \item $f_\epsilon(\{\Psi_\alpha\}_{\alpha \in J})$ is finitely
        $\epsilon$F-satisfiable.
\end{enumerate}
\label{fsat_reduction}
\end{lma}

These two lemmas were used in the example section to express sentences
over different error parameters.

In (\ref{prop:fsat_reduction'}), letting $\epsilon_0 = 0$, we obtain a
computable reduction from finite 0E-validity to finite $\epsilon$E-validity for
any $\epsilon \in (0, 1)$. By corollary
(\ref{fin_eval_Pi1_complete}), we get
\begin{cor}
For any language with an infinite number of unary predicates and
at least one binary predicate,
finite $\epsilon$E-validity is $\Pi^0_1$-hard for rational $\epsilon \in (0, 1)$.
\end{cor}

By theorem (\ref{emodels_finite_sat_is_sigma_1}) for the case of X =
F, finite $\epsilon$F-satisfiability is $\Sigma^0_1$-definable for any language.
By duality, finite $\epsilon$E-validity is $\Pi^0_1$-definable.
Hence, in combination with the above corollary, this implies
\begin{thm}
For any language with an infinite number of unary predicates and
at least one binary predicate,
finite $\epsilon$E-validity is $\Pi^0_1$-complete for rational $\epsilon \in (0,
1)$.

\label{thm:Eval_Pi1_complete}
\end{thm}

Along the same lines, for the countable case, we have
\begin{cor}\label{cor:countabl_epsval_Sigma1_hard}
For any language with an infinite number of unary predicates and
at least one binary predicate,
countable $\epsilon$E-validity is $\Sigma^0_1$-hard for rational $\epsilon \in
(0, 1)$.
\end{cor}

Finally, we mention a theorem of Terwijn.
\begin{prop}[Terwijn \cite{Terwijn_dec}]
Let $\phi$ be a sentence.
$\phi$ is finitely classically valid iff $\phi$ is countably $\epsilon$E-valid
for every $\epsilon > 0$.
\end{prop}
\section{Future Work}
\label{sec:future}

\subsection{The Countable Case}
As displayed by table (\ref{tble:countable_computability}), we still do not
know much about $\epsilon$E-logic over countable $\epsilon$-models.
Most egregiously we have no idea of the computability of its satisfiability
problem.
Looking over the entries of tables (\ref{tble:general_computability}),
(\ref{tble:finite_computability}), and (\ref{tble:countable_computability}),
the pattern seems to favor the possibility of countable
$\epsilon$E-satisfiabilty being $\Pi^0_1$-hard or even complete.
Obviously our proof for the finite case would not carry over, but it is
conceivable that replacing the halting set with a $\Pi^0_1$-complete set would
work out naturally.
Alternatively, we could look at the dual problem of
reducing the halting set to countable $\epsilon$F-validity.

\subsection{Reducing Language Requirement}
In our results, the requirements of an infinite number of unary predicates and
at least three binary predicates are likely not optimal.
In classical first order logic, these requirements can be collapsed to the
single requirement of one binary predicate through graph theoretic or
set theoretic encodings.
However, in $\epsilon$E- and $\epsilon$F- logic, these methods do not seem to
play well with the additional structure of a probability space.
In any case, for our theorems to be more relevant to applications,
the number of unary predicates must be brought down to a finite number, whether
strengthening our undecidability or breaking into decidability.

\subsection{q-Logic and Trees}
We developed q-sentences and other q-concepts only to arrive at results for
$\epsilon$E- and $\epsilon$F-logic, but they can as well be studied on their
own.
In particular, an obvious definition of q-logic would make it a stronger version
of Keisler's probability logic which only allows the quantifiers $\Qr\ge$ and
$\Qr>$.
Keisler's work \cite{keisler_p_quantifiers} can then be applied in most aspects
to such a q-logic.

Similarly, one could investigate the algebraic structure of $\epsilon$E-,
$\epsilon$F-, and q-trees, whose properties we have not fully exploited.
It should not be hard to see that, for a fixed $\epsilon$-model
$\pmd$ and a fixed (q-)sentence $\Phi$, there is a natural partial order and
a join operation on trees $\pmd$ for $\Phi$ of each class.
Instinctively one could ask, under what circumstance does a meet operation
exists?
Deeper research into the semilattice structure of these trees could
reveal information on the computability of $\epsilon$E and $\epsilon$F fragments
not discussed here.

\subsection{Irrational $\epsilon$}
There are several insufficiencies in current techniques with regard to deducing
facts about irrational $\epsilon$s:
\begin{enumerate}[label=(\roman*)]
    \item Terwijn and Kuyper's proofs of the inter-reduction theorems
    fundamentally require the ratio $\epsilon_0/\epsilon_1$ to be rational.
    So while this would imply we have inter-reductions between, say
    $\f 1 {\sqrt{2}}$ and $\f 1 {2\sqrt 2}$, we cannot say much about
    the relative difficulties of $\f 1 2$ and $\f 1 {\sqrt 2}$.
    Surpassing the obstacles to generalize to irrational $\epsilon$ would
    necessitate brand new methods, which
    could also eliminate the infinite unary predicate restriction.
    \item Our proof of the $\Sigma^0_1$-hardness of
    finite $\epsilon$E-satisfiability depends crucially on the ability of $\f 1
    2$E-logic to force the measures of two sets to be equal.
    Without a reduction from $\f 1 2$E- to $\epsilon$E-validity, we cannot
    conclude that $\epsilon$E-validity is $\Sigma^0_1$-hard.
    \item In order for a linear program to be solved in finite time,
    all arithmetic operations over the field generated by its coefficients
    must be total computable.
    This holds for the rational and general number fields,
    but not for the field of computable reals, for which the comparison
    function is not total.
    Thus our proof of decidabilities in monadic relational languages do not
    carry over to the general case.
\end{enumerate}

Fortunately, as our example section illustrate, in many cases only the relative
magnitude of $\epsilon$ matters.
This observation also boosts the likelihood that our computability results
hold for general $\epsilon$ as well.

\subsection{Classical Model Theory Techniques}
Given the applications in section (\ref{ssec:examples_and_applications}) and the
initial motivation of $\epsilon$E-logic, we see that it has an intimate
connection with computational learning theory.
At the intersection of CLT and classical model theory is the concept of
VC dimension
\cite{kearns_CLT}\cite{VCmodel_Laskowski}\cite{VCmodel_I}\cite{VCmodel_II},
which we mentioned briefly.
It could be possible to reconcile the results of these two disciplines in
$\epsilon$E- and $\epsilon$F-logics, to the benefit of all involved.

As our examples also hinted, many techniques in finite model theory could be
converted to versions for our probability logics.
One could also experiment with adding new means of expressions like the BIT
relation, a canonical ordering, counting operators, etc.
An analogue of descriptive complexity could be developed;
given the abundance of probabilistic quantifiers in classes like BPP, PP, PCP,
and so on, it is indeed plausible that one could equate one of these complexity
classes with a description class in $\epsilon$E-logic.

A version of Ehrenfeucht-Fra\"iss\'e games \cite{libkin}, another common tool
in finite model theory, for $\epsilon$E-logic could also have connections to the
malicious advisary learning model
\cite{adversary_Klivans}\cite{adversary_Servedio}.

\subsection{Computational Complexity}
In addition to developing descriptive complexity, one could
also explore typical%
\footnote{to be distinguished from \textit{average-case} complexity.}
-case complexity in $\epsilon$E-logic.

\section{Acknowledgements}

I owe the inspiration for this topic to a conversation with Leslie Valiant,
who provided me with a copy of his work \textit{Robust Logics}
\cite{Valiant_robust_logics}.
His class CS228 at Harvard cultured my appreciation of computational learning
theory and equipped me with the knowledge for much of the material in the
application section.
The first draft of this essay was born as the final project for the course.

I thank Rutger Kuyper for his enthusiastic assistance throughout my research
despite being an ocean away.
My specific interest in $\epsilon$-logic started when I looked for derivative
works of Valiant's \textit{Robust Logics} and discovered Kuyper's thesis.
He helped gather for me all current materials on $\epsilon$-logic and answered
my questions readily, making what could be a rocky journey very smooth.

Finally, and most importantly, I am indebted to Nate Ackerman for the
intellectual conversations, mathematical guidance, and detailed comments on
life, liberty, and an early draft of this paper.
Without him, I would not have seen half of the mathematical world as I have now;
without him, I would not have anywhere to echo my love and amazement of this
magical world;
and without him, I would not be able to amuse myself with his Ackermanic
similes.

\bibliography{ref}

\begin{thebibliography}{10}

\bibitem{VCmodel_I}
M.~{Aschenbrenner}, A.~{Dolich}, D.~{Haskell}, D.~{Macpherson}, and
  S.~{Starchenko}.
\newblock {Vapnik-Chervonenkis density in some theories without the
  independence property, I}.
\newblock {\em ArXiv e-prints}, September 2011.

\bibitem{VCmodel_II}
M.~{Aschenbrenner}, A.~{Dolich}, D.~{Haskell}, D.~{Macpherson}, and
  S.~{Starchenko}.
\newblock {Vapnik-Chervonenkis density in some theories without the
  independence property, II}.
\newblock {\em Notre Dame Journal of Formal Logic}, 54(3-4):311--363, 2013.

\bibitem{birkhoff}
Garrett Birkhoff.
\newblock {\em Lattice Theory}.
\newblock American Mathematical Society, revised ed. edition, 1948.

\bibitem{Bogachev}
Vladimir~I. Bogachev.
\newblock {\em Measure Theory}.
\newblock Springer, 2007.

\bibitem{kearns_agnostic}
Michael~J. Kearns, Robert~E. Schapire, and Linda~M. Sellie.
\newblock Toward efficient agnostic learning.
\newblock {\em Machine Learning}, 17(2-3):115--141, 1994.

\bibitem{kearns_CLT}
Michael~J. Kearns and Umesh~Virkumar Vazirani.
\newblock {\em An Introduction to Computational Learning Theory}.
\newblock MIT Press, Jan 1994.

\bibitem{keisler_p_quantifiers}
H.~J. Keisler.
\newblock {\em Chapter XIV: Probability Quantifiers}, volume~8 of {\em
  Perspectives in Mathematical Logic}, pages 507--556.
\newblock Springer-Verlag, New York, 1985.

\bibitem{adversary_Klivans}
Adam~R. Klivans, Philip~M. Long, and Rocco~A. Servedio.
\newblock Learning halfspaces with malicious noise.
\newblock {\em J. Mach. Learn. Res.}, 10:2715--2740, December 2009.

\bibitem{attribute_efficient_klivans}
Adam~R. Klivans and Rocco~A. Servedio.
\newblock Toward attribute efficient learning of decision lists and parities.
\newblock In John Shawe-Taylor and Yoram Singer, editors, {\em Learning
  Theory}, volume 3120 of {\em Lecture Notes in Computer Science}, pages
  224--238. Springer Berlin Heidelberg, 2004.

\bibitem{sat_is_sigma_1_1}
Rutger Kuyper.
\newblock Computational aspects of satisfiability in probability logic.
\newblock {\em to appear in Mathematical Logic Quarterly}.

\bibitem{probability_logic}
Rutger Kuyper.
\newblock Probability logic.
\newblock Master's thesis, Radboud University Nijmegen, 2011.

\bibitem{val_is_pi_1_1}
Rutger Kuyper.
\newblock Computational hardness of validity in probability logic.
\newblock In Sergei Artemov and Anil Nerode, editors, {\em Logical Foundations
  of Computer Science}, volume 7734 of {\em Lecture Notes in Computer Science},
  pages 252--265. Springer Berlin Heidelberg, 2013.

\bibitem{model_theory_of_e_logic}
Rutger Kuyper and Sebastiaan~A. Terwijn.
\newblock Model theory of measure spaces and probability logic.
\newblock {\em The Review of Symbolic Logic}, 6:367--393, 9 2013.

\bibitem{VCmodel_Laskowski}
Michael~C. Laskowski.
\newblock Vapnik-chervonenkis classes of definable sets.
\newblock {\em Journal of the London Mathematical Society}, s2-45(2):377--384,
  1992.

\bibitem{libkin}
Leonid Libkin.
\newblock {\em Elements of Finite Model Theory}.
\newblock Springer, 2004.

\bibitem{mackay}
David J.~C. MacKay.
\newblock {\em Information Theory, Inference, and Learning Algorithms}.
\newblock Cambridge University Press, 2003.

\bibitem{LP_complexity_survey}
Nimrod Megiddo.
\newblock On the complexity of linear programming.
\newblock In Truman~Fassett Bewley, editor, {\em Advances in Economic Theory},
  pages 225--268. Cambridge University Press, 1987.
\newblock Cambridge Books Online.

\bibitem{russel_norvig}
Stuart~J. Russel and Peter Norvig.
\newblock {\em Artificial Intelligence: A Modern Approach}.
\newblock Prentice Hall, 2010.

\bibitem{schrijver}
Alexander Schrijver.
\newblock {\em Theory of Linear and Integer Programming}.
\newblock Wiley, 1998.

\bibitem{adversary_Servedio}
Rocco~A. Servedio.
\newblock Smooth boosting and learning with malicious noise.
\newblock {\em J. Mach. Learn. Res.}, 4:633--648, December 2003.

\bibitem{Soare}
Robert~I. Soare.
\newblock {\em Recursively Enumerable Sets and Degrees}.
\newblock Springer, 1987.

\bibitem{Terwijn_intro}
Sebastiaan~A. Terwijn.
\newblock Probabilistic logic and induction.
\newblock {\em Journal of Logic and Computation}, 15(4):507--515, 2005.

\bibitem{Terwijn_dec}
Sebastiaan~A. Terwijn.
\newblock Decidability and undecidability in probability logic.
\newblock In Sergei Artemov and Anil Nerode, editors, {\em Logical Foundations
  of Computer Science}, volume 5407 of {\em Lecture Notes in Computer Science},
  pages 441--450. Springer Berlin Heidelberg, 2009.

\bibitem{Turing}
A.~M. Turing.
\newblock On computable numbers, with an application to the
  entscheidungsproblem.
\newblock {\em Proceedings of the London Mathematical Society},
  s2-42(1):230--265, 1937.

\bibitem{valiant_PAC}
L.~G. Valiant.
\newblock A theory of the learnable.
\newblock {\em Commun. ACM}, 27(11):1134--1142, November 1984.

\bibitem{Valiant_robust_logics}
Leslie~G. Valiant.
\newblock Robust logics.
\newblock In {\em Proceedings of the Thirty-first Annual ACM Symposium on
  Theory of Computing}, STOC '99, pages 642--651, New York, NY, USA, 1999. ACM.

\end{thebibliography}
\bibliographystyle{plain}
\end{document}